\documentclass[11pt,letterpaper]{article}
\usepackage[utf8]{inputenc}
\usepackage{amsmath}
\usepackage{amsfonts}
\usepackage{amssymb}
\usepackage{amsthm}
\usepackage{todonotes}
\usepackage{array}
\newcolumntype{?}{!{\color{lightgray}\vrule width .01pt}}
\usepackage[left=2.6cm,right=2.6cm,top=2.6cm,bottom=2.6cm]{geometry}
\usepackage{enumitem}
\usepackage{gensymb}
\usepackage{graphicx}
\usepackage{verbatim}
\usepackage{dsfont}
\usepackage{bbm}
\usepackage{tikz}
\usepackage{algpseudocode}
\newcommand{\mb}{\mathbb}

\newcommand{\poly}{\operatorname{poly}}
\newcommand{\polylog}{\operatorname{polylog}}
\newcommand{\E}{\mathbb{E}}

\allowdisplaybreaks
\usetikzlibrary{positioning}
\usepackage[colorlinks=false,urlbordercolor=white]{hyperref}
\newcommand{\say}[1]{\text{\hspace{1.5cm}#1}}

\tikzset{sgplattice/.style={inner sep=1pt,norm/.style={red!50!blue},char/.style={blue!50!black},
  lin/.style={black!50}},cnj/.style={black!50,yshift=-2.5pt,left=-1pt of #1,scale=0.5,fill=white}}

\newtheorem{lemma}{Lemma}[section]
\newtheorem*{lemma*}{Lemma}
\newtheorem{corollary}[lemma]{Corollary}

\newtheorem{definition}[lemma]{Definition}

\makeatother

\usepackage{thmtools}
\usepackage{thm-restate}

\usepackage{color}

\author{Tolson Bell\thanks{thbell@alumni.cmu.edu. Research supported in part by NSF GRFP Grant DGE 2140739.}~~and William Kuszmaul\thanks{william.kuszmaul@gmail.com. Research supported in part by NSF grant CCF-2504471 and by a Jane Street grant.}\\Carnegie Mellon University\\Pittsburgh, PA 15213\\U.S.A.}
\title{Fast Insertion for Bucketized Cuckoo Hashing}
\date{}
\begin{document}
\maketitle
\thispagestyle{empty}
\begin{abstract}
Bucketized cuckoo hashing is a practically efficient hash table scheme in which each object $u$ is stored in one of two buckets $h_1(u), h_2(u)$ of capacity $\ell$. For any bucket size $\ell\in\mb{N}$, there is a threshold $\epsilon^*(\ell)=(2/e)^\ell\mathrm{poly}(\ell)$ for which there exists a way to fill the hash table to any load factor less than $1-\epsilon^*$ with low probability of an error. Queries and deletions only need to check two buckets to find whether an object exists.

Our contribution is to give a new insertion procedure for bucketized cuckoo hashing. For any $\delta\in[.99^\ell,1]$, our algorithm can fill the hash table to load factor $1-\epsilon=1-(1+\delta)(\epsilon^*)$ with an expected run time of $O(\delta^{-1}(\epsilon^*)^{-1})$ per insertion. This gives the first $\mathrm{poly}(\epsilon^{-1})$ insertion time bound, and the first $f(\epsilon^{-1})$ time bound for load factors that are very close to the optimal threshold. Additionally, our algorithm (which can be viewed as a variation of the classic random-walk algorithm) comes with a very strong amortized guarantee: it performs $O(1)$ amortized expected evictions per insertion. Furthermore, we show that the traditional random-walk algorithm cannot match this guarantee.

Finally, our insertion protocol also comes with the feature that, for any key $u$ in the hash table, the query algorithm can \emph{guess} which of the two bins $h_1(u), h_2(u)$ the key $u$ is in with probability $1 - o(1)$ of being correct. Thus positive queries can complete in $1 + o(1)$ expected bin accesses. 
\end{abstract}
\newpage\clearpage\pagenumbering{arabic}
\section{Introduction}\label{introsection}

\emph{Bucketized cuckoo hashing} \cite{DW07} is an open-addressed hash table that uses two hash functions $h_1, h_2: U \rightarrow [m]$ to assign objects $x \in U$ to bins $b \in [m]$ of size $\ell$. By rearranging elements over time, the hash table maintains the invariant that, if an element $x$ is present, it is in one of bins $h_1(x)$ or $h_2(x)$. 

In practice, even relatively small buckets allow the hash table to operate at very high load factors \cite{DW07,li2014algorithmic,fan2014cuckoo,kuszmaul2016fast}. For query-heavy workloads, the result is a nearly ideal hash table: each query completes with at most two cache misses and in worst-case time $O(\ell)$. The main drawback is the insertion time, which can become quite slow as the hash table reaches higher load factors \cite{li2014algorithmic,fan2014cuckoo,kuszmaul2016fast}.

The situation from a theory perspective is quite similar. Given a bucket size of $\ell$, the maximum load factor $1 - \epsilon^*$ at which a bucketized cuckoo hash table can operate satisfies $\epsilon^* = (2/e)^{\ell + o(\ell)}$ as $\ell\rightarrow\infty$ \cite{DW07,binsthreshold2, binsthreshold1, Lelarge}. This means that the \emph{query time} is $O(\log (({\epsilon^*})^{-1}))$, which is far better than one obtains with other common schemes such as linear probing, uniform probing, etc. \cite{knuthvol3}. 

But the insertion time is much trickier to reason about. A classic result by Dietzfelbinger and Weidling \cite{DW07} shows that it is possible to support $1/\epsilon^{O(\ell)}$ expected-time insertions at any load factor $1 - \epsilon$ satisfying $\epsilon \ge e^{-\ell/16}$ -- for small $\epsilon$, say $\epsilon = e^{-\Theta(\ell)}$, this bound becomes $\epsilon^{-O(\log \log \epsilon^{-1})}$. Whether the insertion time can be improved to $O(\epsilon^{-1})$, or even to $\poly(\epsilon^{-1})$, has remained open \cite{DW07,FriezePetti,KM25} (although it is conjectured that several natural insertion algorithms should do this \cite{DW07}). And, when $\epsilon$ gets closer to $\epsilon^*$, it remains unknown even whether any bound of the form $f(\epsilon^{*})$ should be possible.

The main contribution of this paper is a new insertion algorithm that supports $\epsilon$ as small as $(1 + O(1)) \cdot \epsilon^*$ with a worst-case expected insertion time of $O(\epsilon^{-1} \cdot \ell)$. When $\epsilon$ is small, satisfying $\epsilon = e^{-\Theta(\ell)}$, even the $\ell$ term disappears, resulting in a time bound of $O(\epsilon^{-1})$. The insertion algorithm can also support $\epsilon$ of the form $(1 + \delta) \epsilon^*$, where $\delta \in [0, 1]$ can be as small as $e^{-\Omega(\ell)}$, and where the expected insertion time becomes $O(\epsilon^{-1} \delta^{-1})$.

In addition to being time efficient, our algorithm comes with several surprising features. Even though it can be viewed as a variation of the classical ``random-walk'' insertion scheme, it actually performs \emph{provably} fewer total evictions in order to get to high load factors. And, even though the primary goal of the algorithm is to support fast insertions, it also comes with an inadvertent advantage for queries: for any element $x$ in the hash table, if $\ell = \omega(1)$, then the query algorithm can \emph{predict} with probability $1 - o(1)$ which of the bins $h_1(x)$ or $h_2(x)$ contains $x$, at any given moment. This means that, in many natural parameter regimes, successful queries can complete with $1 + o(1)$ expected cache misses.

Throughout the paper, we assume that $\ell$ is at least a sufficiently large constant, and that $\ell \le 1.5\ln(m)$. (The reason for the upper bound is to avoid issues that arise once $\E[\epsilon^*]$ gets to about $1/\sqrt{m}$ or smaller.) 

\paragraph{What makes insertions hard?} Before diving into the results, let us take a moment to understand what makes bucketized cuckoo hashing difficult to analyze in the first place. 

Suppose we wish to analyze the basic \emph{random-walk} insertion strategy: to insert an element $x$ into a full bin $h_i(x)$, we simply \emph{evict} a random $y$ from the bin, and recursively insert $y$ into its other bin $b \in \{h_1(y), h_2(y)\} \setminus h_i(x)$. 

The main difficulty is the issue of \emph{recycled randomness}. The first time that an element $x$ tries one of its bins $h_i(x)$, that bin is truly random. But, if the element $x$ is evicted to its other bin $h_{\overline{i}}(x)$, and is then eventually evicted \emph{back} to $h_i(x)$, we can no longer think of $h_i(x)$ as random. Intuitively, the fact that the bin already evicted $x$ in the past makes the same bin more likely to perform additional evictions in the future. 

One way to combat recycled randomness is to allow for $d > 2$ hash functions, so that we can perform more ``truly random'' probes per element. In the case where $d$ is large and $\ell = 1$, this approach is known as $d$-ary cuckoo hashing \cite{FPSS05,bell20241,KM25}, and efficient insertions become possible \cite{KM25}. However, as noted by Kuszmaul and Mitzenmacher \cite{KM25}, the issue of reappearance dependencies seems to be more substantial for the bucketized case, since it is no longer acceptable to ``waste hashes'': one must make good use out of \emph{both} hash functions if one wishes to achieve load factors close to $1 - \epsilon^*$. 

The issue of recycled randomness is closely related to another issue, which is that of a \emph{hotspot}. If a bin $b$ gets many evictions, then the bins $b'$ that it evicts to will also tend to get more evictions, and so on. If, eventually, these bins have tried both hashes for all their elements, then we can end up in a situation where each insertion goes through a long eviction chain of bins that have all already used up all of their ``fresh randomness''. These hotspots are exacerbated by a natural type of positive feedback loop: the more elements a bin has evicted in the past, the more likely they are to return, and the more likely the bin is to evict its other elements, eventually using up all its ``fresh randomness''. Avoiding these types of hotspots is crucial for achieving good insertion times.

\paragraph{Part 1: A Warmup Algorithm. } We begin the paper in Section \ref{firstalgsection} by presenting a relatively simple warmup algorithm that shows how to get around (at least, in some parameter regimes) the issues of recycled randomness and hotspots. The warmup algorithm is already able to achieve an expected insertion time of $O(\epsilon^{-1} \ell)$, but with the restriction that it only supports $\epsilon$ down to $O({\epsilon^*}e^{\ell^{.51}})$ (Theorem \ref{thmfirst}). 

At a high level, the warmup algorithm makes two modifications to the random-walk insertion strategy. First, to combat recycled randomness, it prioritizes evicting elements $y$ whose other hash has never yet been tried. Second, whenever the algorithm \emph{does} return an element $y$ to a bin $h_i(y)$ that had already contained $y$ in the past, the algorithm \emph{commits} to placing $y$ in the top slot of the bin. This may seem like a strange optimization, since it is actually a \emph{restriction} on the algorithm's ability to choose which element to evict. But this restriction ends up being essential to the analysis: it guarantees that the bin contains \emph{at most} one such $y$ at any given moment, preventing positive feedback loops in which the bin becomes more and more saturated by elements who have already been evicted from the bin in the past. 

With these modifications in place, we show that there is a clean path to analyzing the algorithm. Specifically, we formalize the idea that there are no large hotspots, proving that contiguous sequences of evictions that use recycled randomness are $O(1)$ length in expectation (this corresponds to a path in the ``corruption graph'' in the analysis); then, we use the fresh randomness that we encounter whenever we are not in a hotspot to prove that each insertion takes expected time $O(\epsilon^{-1}\ell)$.

An interesting feature of the algorithm (which also holds for the second algorithm) is that the \emph{total} number of evictions over all insertions is expected $O(\ell m)$. Notice that there are $\Theta(\ell m)$ total insertions, so this means that the amortized expected number of \emph{evictions} per operation is $O(1)$. We note that, although our algorithm can be viewed as a variation of the random-walk algorithm, the random-walk algorithm provably does not achieve this amortized bound: We argue in Section \ref{sec:separation} that, if one wishes to get to a load factor of the form $1-(2/e)^{\ell-o(\ell)} = 1 - (\epsilon^*)^{1 - o(1)}$, then the random-walk strategy performs at least $\omega(1)$ amortized expected evictions per insertion. Thus the optimizations in the algorithm make a difference not just analytically but also in terms of the actual behavior of the algorithm.

\paragraph{Part 2: The Full Result. }
Having completed the warmup, in Section \ref{secondalgsection}, we present the main result of the paper: an insertion algorithm that supports expected insertion time $O(\epsilon^{-1} \ell)$ for any $\epsilon \in (1 + \Omega(1))\epsilon^*$. (When $\epsilon$ is small, satisfying $\epsilon \le e^{-\Omega(\ell)}$, the insertion time improves slightly to $O(\epsilon^{-1})$.) In fact, the algorithm even supports $\epsilon$ as small as $(1 + \delta) \epsilon^*$, where $\delta \in (0, 1)$ can itself be as small as $e^{-\Omega(\ell)}$, and where the expected insertion time becomes $O(\epsilon^{-1} \delta^{-1})$ (Theorem \ref{secondalgthm}). 

We remark that the time bound of $O(\epsilon^{-1} \delta^{-1})$ comes from the following subtle phenomenon: at load factor $1 - (1 - \delta)\epsilon^*$, the hash table is so full that the vast majority (roughly a $1- \delta$ fraction) of remaining free slots are \emph{unreachable} (every element $u$ capable of using the free slot is already in the same bin as the free slot). This means that  the \emph{effective} number of remaining free slots is roughly $\delta \epsilon$, hence the time bound of $O(\delta^{-1} \epsilon^{-1})$. We remark that this phenomenon is one that no previous insertion algorithm has had to encounter because no algorithm has been able to support load factors so close to $1 - \epsilon^*$. 

The basic idea of the algorithm is to mix the warmup algorithm with a very \emph{shallow} breadth-first-search step. When evicting from a bin, we first check whether any of the elements $y$ in the bin can be evicted to another bin that has a free slot. If so, we evict the element $y$ to that bin and continue the insertion process. If not, we continue with essentially the same insertion process as used by the warmup algorithm. 

This shallow BFS step allows us to argue that, by the time we have visited each bin $b$, say, $1.6 \ell$ times on average, we have most likely already made ``good use'' of all the hashes of all the elements in the bin. That is, if the bin ever had any elements $y$ whose other hash is a bin with a free slot, we have most likely already sent $y$ there. This, in turn, allows us to argue that the hash table reaches very high load factors even \emph{well before} it has evicted every element once. This means that insertions get to experience a hash table in which there are still many ``pretty fresh'' probes to explore (elements whose second hash has been explored by the BFS step but never by the random-walk step). These ``pretty fresh'' probes, in turn, allow the insertion to efficiently find elements that have never yet been probed even by the BFS step, and those elements have a reasonable chance of taking the insertion to a bin with a free slot. By formalizing this high-level intuition into the analysis, we are able to prove that the expected insertion time is $O(\epsilon^{-1}\delta^{-1})$.

Interestingly, the analysis also naturally yields a precise approximation for $\epsilon^*$, up to $1 \pm o(1)$ factors. Namely, we get that
\begin{equation}\label{eq:epsilonstar}
\epsilon^* = (1 \pm O(1/\ell)) \cdot \frac{2\cdot(2/e)^\ell}{\ell^{1.5}\sqrt{2\pi}},
\end{equation}
or even more precisely that $\epsilon^* \approx (1 \pm e^{-\Omega(\ell)}) \cdot \frac{1}{\ell} \cdot \E_{X\sim\operatorname{Poisson}(2\ell)}[\max(0, \ell - X)]$. This approximation captures the fact that, for the most part, the load threshold is bottlenecked by the existence of bins $b$ for which there are fewer than $\ell$ \emph{distinct elements} $x$ satisfying $h_1(x) = b$ or $h_2(x) = b$. Such bins necessarily must have free slots in any hash-table configuration. Although an exact calculation of $\epsilon^*$ is known \cite{Lelarge,FKPloadthresholds}, the observation that, up to $1 \pm o(1)$ factors, $\epsilon^*$ is approximated by \eqref{eq:epsilonstar} appears to be a new observation that naturally pops out of our analysis.

\paragraph{Optimizing Queries. }In Section \ref{querysection}, we turn our attention to queries. We show that, with a simple modification, both of our algorithms allow for one to predict which bin $h_1(x), h_2(x)$ a given key $x$ will be in at any given moment. And, assuming $\ell = \omega(1)$, this prediction will be correct with probability $1 - o(1)$.

Specifically, we assign each key $x$ a random \emph{eviction priority} $p(x) \in [0,1]$, and when breaking ties to decide who to evict from a bin, we always evict the element $y$ with the highest eviction priority. This results in the following effect: at any given moment, there exists a threshold $\tau$ such that, for a given element $x$ in the hash table, the question of which bin $h_1(x)$ or $h_2(x)$ the element is in is, with probability at least $1 - o (1)$, correctly predicted by whether $p(x) < \tau$. 

As a consequence, for \emph{successful} queries (queries to elements that are present), the algorithm needs only examine $1 + o(1)$ bins in expectation. If we consider such a data structure in the external-memory model, with blocks of size $B = \ell$, then the expected number of \emph{cache misses} per successful query is also $1 + o(1)$, while the \emph{worst-case} number of cache misses remains $2$.

\paragraph{Related work. }Cuckoo hashing (with buckets of size 1) was introduced and analyzed by Pagh and Rodler in a highly influential 2001 paper \cite{PR01} (building on a static construction by Pagh \cite{pagh2001cell} in the same year). Similar data structures were independently proposed in earlier works, including a 1981 technical report by Lyon \cite{Lyon1981AlternationTree,NTRL1981LyonAlternationTree} (also published as a journal article in 1985 \cite{lyon1985achieving}), which described bucketized cuckoo hashing with buckets of size $\ell \ge 2$; and a Sun Microsystems patent filed in 1997 by Hagersten and  Hill \cite{HagerstenHill1999ScalableSharedMemory,hagersten2001shared}, which described the non-bucketized variant.

As discussed earlier, bucketized cuckoo hashing \cite{lyon1985achieving, DW07} has been studied in both the static \cite{DW07,binsthreshold2, binsthreshold1, Lelarge} and dynamic \cite{DW07,FriezePetti} settings. In the dynamic setting, it is conjectured that both breadth-first-search and random-walk insertions should achieve $O(\epsilon^{-1})$ expected time (for $\epsilon$ close to $\epsilon^*$) , but the best known bounds have remained either super-polynomial in $\epsilon^{-1}$ (for breadth-first search \cite{DW07}) or limited only to very large $\epsilon$ (for random-walk insertions \cite{FriezePetti}, with $\epsilon \ge \tilde{\Omega}(1/\sqrt{\ell})$). In the static setting, the maximum load factor $1 - \epsilon^*$ at which a bucketized cuckoo hash table can operate satisfies $\epsilon^* = (2/e)^{\ell + o(\ell)}$ as $\ell\rightarrow\infty$ \cite{DW07,binsthreshold2, binsthreshold1, Lelarge}. Researchers have also studied how the critical load threshold behaves for variations of bucketized cuckoo hashing in which buckets are allowed to overlap \cite{lehman20093, walzer2023load}.

Another widely-studied variant of cuckoo hashing is $d$-ary cuckoo hashing \cite{FPSS05}, where each element hashes to $d > 2$ slots (as opposed to $2$ buckets). There has been a large body of work on analyzing the insertion time of $d$-ary cuckoo hashing under different algorithms (especially random walk insertion) \cite{PolylogAlan,FPS13,wearmin,FriezeJohansson,newinsertion,Walzer,bell20241,KM25}. The state of the art \cite{KM25} allows for $O(\epsilon^{-1})$ time insertions for $\epsilon$ within a constant-factor of the optimal $\epsilon^* = e^{-d \pm o(1)}$ (with $\epsilon^*$ defined, now, for $d$-ary cuckoo hashing).

An advantage of bucketized cuckoo hashing over $d$-ary cuckoo hashing is the data locality \cite{li2014algorithmic,fan2014cuckoo}. While each object has $2\ell$ possible slots, those slots come in two contiguous blocks of memory. In practice, this allows bucketized cuckoo hashing to support queries in at most two cache misses \cite{li2014algorithmic,fan2014cuckoo}.

Bucketized and $d$-ary cuckoo hashing can be combined, with $d$ hash functions and buckets of size $\ell$. This variation has been studied primarily in the offline setting, where the goal is to determine the critical load threshold for any given number $d$ of hash functions and bin size $\ell$. After a long line of work \cite{GWloadthresholds,xorsatthresholds,FP12,FM12,Lelarge}, it is now known how to determine this threshold for arbitrary $(d, \ell)$ \cite{FKPloadthresholds}. 

\section{Preliminaries}\label{epsstarsection}

\subsection{Definitions and Notation}\label{prelimsubsect}

We will sometimes refer to the hash table objects $x\in U$ as \underline{balls}, as we are inspired by traditional ``balls into bins'' analyses. We say that each bin $B\in[m]$ consists of $\ell$ slots, so our hash table has $\ell m$ slots and thus could possibly accommodate up to $\ell m$ balls. Slots that do not currently have any ball stored in them are referred to as \underline{empty} or \underline{free}. Bins that have fewer than $\ell$ balls stored in them are referred to as \underline{non-full}.

For our algorithms and analysis, we imagine that there are $\ell m$ total balls that we have waiting to be inserted (that is, the same number of balls as table slots). Our algorithms insert these balls one by one until a prescribed stopping time is reached. We call the entire process of inserting these balls \underline{the insertion process}. We consider the $\ell$ slots of each bin to be ordered from first to last within the bin (and keep this order throughout the insertion process).

\begin{definition}\label{epsstardef}
Let $\epsilon^*$ be the random variable such that no valid assignment of balls to bins exists when there are $\epsilon^*\ell m$ balls (of the $\ell m$ total balls) left (and thus $\epsilon^*\ell m$ empty slots left).
\end{definition}  This definition is different from that in our introduction and some other works, which use $\epsilon^*$ to denote the mean value of this random variable, which we now call $\mb{E}[\epsilon^*]$. Prior work has shown that for any fixed $\ell$, we have that $\epsilon^*$ is a sharp threshold under the random graphs definition as $m$ tends towards infinity \cite{FKPloadthresholds}. The existence and formula for $\mb{E}[\epsilon^*]$ is known for any fixed $\ell$ and $m\rightarrow\infty$ \cite{Lelarge,FKPloadthresholds}.

We assume throughout the paper that $m\rightarrow\infty$, that $\ell$ is at least a sufficiently large constant, and that $\ell\le 1.5\ln(m)$. As we will show that $\epsilon^*=(2/e)^\ell\poly(\ell)$, this restriction will imply that $\epsilon^*=\Omega(m^{-.47})$.

The definition we use for \underline{with high probability in $m$} is that for any $c\in\mb{N}$, there is a $C$ (depending on $c$) such that the probability is at least $1-Cm^{-c}$.

\subsection{Bounding the Optimal Load Factor $\epsilon^*$}\label{boundingsubsect}

If a given bin receives only $j$ hashes out of our collection of $\ell m$ balls (so $2\ell m$ total hashes), we have that at least $\max(\ell-j,0)$ slots must remain empty under any valid assignment, which we call the \underline{leftover slots} for that bin.
\begin{definition}\label{epstildedef}
Let $\tilde\epsilon$ be the random variable such that there are $\tilde\epsilon\ell m$ total leftover slots from our collection of $\ell m$ balls, that is, $\tilde\epsilon\ell m=\sum_{\text{bins }B}\max(0,\ell-(\text{\# hashes to }B))$.
\end{definition}
We see that $\epsilon^*\ge\tilde\epsilon$ always, as if there are $\tilde\epsilon\ell m$ leftover slots, then there is not any valid assignment of $(1-\tilde\epsilon)\ell m + 1$ balls.

Using the fact that binomial random variables behave essentially like Poissons, we can obtain the following simple approximations for the mean and concentration bounds for $\tilde\epsilon$.
\begin{restatable}{lemma}{expepstilde}\label{expepstilde}
Assuming that $\ell=O(\log m)$, we have that \[\mb{E}[\tilde\epsilon]=(1 \pm O(1/m)) \cdot\frac 1\ell\cdot \E_{X\sim\operatorname{Poisson}(2\ell)}[\max(0, \ell - X)]=(1 \pm O(1/\ell)) \cdot \frac{2\cdot(2/e)^\ell}{\ell^{1.5}\sqrt{2\pi}}.\]
\end{restatable}
\begin{restatable}{lemma}{boundingepstilde}\label{boundingepstilde}
Assuming that $\ell=O(\log m)$, we have that, with high probability in $m$,\[\tilde\epsilon=\mb{E}[\tilde\epsilon]\pm m^{-.49}\]
\end{restatable}
The proofs of Lemma \ref{expepstilde} and \ref{boundingepstilde} will be given in the Appendix \ref{app:foo}.

Issues with our analysis arise if $\tilde\epsilon$ is not sufficiently concentrated, which we see by Lemma \ref{boundingepstilde} might occur if the hash table stores $\gtrsim\ell m-O(\sqrt{m})$ balls. To avoid this, we want to assume that $\ell\le c\ln(m)$ for some $c<\frac{.5}{\ln(e/2)}\approx 1.63$. To give ourselves a bit more wiggle room, we assume that $\ell\le 1.5\ln(m)$. As we proved in Lemma \ref{expepstilde} that $\mb{E}[\tilde\epsilon]=(2/e)^\ell\poly(\ell)$, this in particular implies that for the rest of the paper we can assume that $\mb{E}[\tilde\epsilon]\ge(2/e)^{1.5\ln(m)}\poly(\log m)=\Omega(m^{-.47})$. We then also have that $m^{-.49}$ is $O(\mb{E}[\tilde\epsilon]/\ell)$, and so Lemmas \ref{expepstilde} and \ref{boundingepstilde} give us that, when $\ell\le 1.5\ln(m)$, then with high probability in $m$ we have $\mb{E}[\tilde\epsilon]=\Omega(m^{-.47})$ and \[\epsilon^*\ge\tilde\epsilon=\mb{E}[\tilde\epsilon]\pm m^{-.49}=(1 \pm O(1/\ell)) \cdot \frac{2\cdot (2/e)^\ell}{\ell^{1.5}\sqrt{2\pi}}.\]

In Section \ref{firstalgsection}, we show that our first algorithm successfully inserts at least $(1-4e^{\ell^{.51}}\tilde\epsilon)\ell m$ elements with probability at least $1-O(\ell^{-6}m^{-1})$, in which case we have \[4e^{\ell^{.51}}\tilde\epsilon=\left(\frac 2e\right)^{\ell-o(\ell)}\ge\epsilon^*.\] In Section \ref{secondalgsection}, and in particular Lemmas \ref{adapted2} and \ref{epsstarsecond}, we get even more precise, and show that our second algorithm succeeds at inserting at least $(1-(1+.99^\ell)\tilde\epsilon)\ell m$ elements with probability at least $1-O(\ell^{-6}m^{-1})$, showing that with probability at least $1-O(\ell^{-6}m^{-1})$ we have \[(1+.99^\ell)\tilde\epsilon=(1\pm O(1/\ell))\cdot\frac{2\cdot (2/e)^\ell}{\ell^{1.5}\sqrt{2\pi}}\ge\epsilon^*\ge\tilde\epsilon=(1 \pm O(1/\ell)) \cdot \frac{2\cdot (2/e)^\ell}{\ell^{1.5}\sqrt{2\pi}}.\]
In conclusion, Lemmas \ref{adapted2} and \ref{epsstarsecond} give the following corollary:
\begin{corollary}\label{boundingepsstar}
Assume that $\ell\le 1.5\ln(m)$. With probability $1-O(\ell^{-6}m^{-1})$, we have that \[
\mb{E}[\tilde\epsilon]-m^{.51}\le\tilde\epsilon\le\epsilon^*\le(1+.99^\ell)\tilde\epsilon\le(1+.99^\ell)\mb{E}[\tilde\epsilon]+m^{.51}
\]
and thus
\[
\epsilon^*=(1 \pm O(1/\ell)) \cdot \frac{2\cdot (2/e)^\ell}{\ell^{1.5}\sqrt{2\pi}}.
\]Or, even more precisely, \[
\epsilon^*=(1 \pm e^{-\Omega(\ell)})\E[\tilde\epsilon]= (1 \pm e^{-\Omega(\ell)})\cdot\frac 1\ell\cdot \E_{X\sim\operatorname{Poisson}(2\ell)}[\max(0, \ell - X)].
\]
\end{corollary}

The fact that binomial random variables can be closely approximated by Poisson random variables will be helpful (implicitly) throughout the paper. To capture this formally, we use the following two basic lemmas.
\begin{restatable}{lemma}{poisson}\label{poisson}
If $c\in(0,1)$, then \[\mb{P}(\mathrm{Bin}(c\ell m,1/m)\ge\ell-1)\le c^{-1}e^{-\ell(c-1-\ln(c))}\le c^{-1}e^{-\ell(1-c)^2/2}.\]
If $c\in(1,1.7)$, then \[\mb{P}(\mathrm{Bin}(c\ell m,1/m)\le\ell)\le e^{-\ell(c-1-\ln(c))}\le e^{-\ell(c-1)^2/3}.\]
\end{restatable}
\begin{restatable}{lemma}{genpoisson}\label{genpoisson}
Let $c_1\in (0,1)$. If $c_2\in(0,c_1)$, then \[\mb{P}(\mathrm{Bin}(c_2\ell m,1/m)\ge c_1\ell)\le e^{-\ell(c_1\ln(c_1/c_2)+c_1-c_2)}\le e^{-\ell(c_1-c_2)^2/2}.\]
If $c_2\in(c_1,1)$, then \[\mb{P}(\mathrm{Bin}(c_2\ell m,1/m)\le c_1\ell)\le e^{-\ell(c_1\ln(c_1/c_2)+c_1-c_2)}\le e^{-\ell(c_1-c_2)^2/2}.\]
\end{restatable}
Note that Lemmas \ref{poisson} and \ref{genpoisson} do allow $c$ to depend on $\ell$. For instance, we will later apply Lemma \ref{poisson} with $c=1-.5\ell^{-.49}<1$. The proofs of Lemmas \ref{poisson} and \ref{genpoisson} will be given in Appendix \ref{app:foo}.

\section{The First Algorithm: Prioritizing Unrevealed Evictions}\label{firstalgsection}
In this section, we will describe an insertion algorithm, and prove that it gives the following guarantees:

\begin{restatable}{theorem}{thmfirst}\label{thmfirst}
Assume $\ell\le 1.5\ln(m)$ is at least a sufficiently large positive constant. With probability at least $1-O(\ell^{-6}m^{-1})$ over the insertion process, our algorithm succeeds in inserting $(1-\epsilon)\ell m$ objects for $\epsilon\le 4e^{\ell^{.51}}\epsilon^*\le\left(\frac 2e\right)^{\ell-o(\ell)}$. For an insertion at load factor $1 - \epsilon$, the expected number of evictions by the algorithm is $O(\epsilon^{-1})$, and the expected insertion time is $O(\epsilon^{-1} \ell)$. Finally, the expected total number of evictions that our algorithm performs across all insertions is $O(\ell m)$.
\end{restatable}

We remark that the $O(\epsilon^{-1})$ bound on number of evictions is actually not tight when $\epsilon = e^{-\Omega(\ell)}$. A tighter analysis can bring this bound down to $O(\epsilon^{-1} / \ell)$, and can likewise bring the insertion time down to $O(\epsilon^{-1})$. For simplicity, we defer this type of tighter analysis to our treatment of our second algorithm in Section \ref{secondalgsection} (Theorem \ref{secondalgthm}).

\subsection{Algorithm Definitions and Basic Properties}
We now describe our first insertion algorithm. One major idea of the algorithm is to first place balls in their first hash position if possible, without revealing their second hash position. Leaving the second hash unrevealed means that our analysis can defer the randomness for later and keep it as a random hash. Furthermore, when evicting a ball from a bin, we prefer to evict these balls whose second hash is unrevealed, rather than those balls who have already been to both of their hash positions.

When inserting an object $x_0$, we first reveal its first hash, $h_1(x_0)$, and start with $B$ initialized to equal $B_0=h_1(x_0)$. The key component of the algorithm is to decide which ball $x'$ should be evicted from $B$ when a ball $x$ is inserted into $B$. We then change $B$ to be the other hash location of $x'$, which becomes our new $x$ to insert. We evict in the following order:
\begin{enumerate}
    \item If $B$ has an empty slot, place $x$ into the first empty slot of $B$.
    \begin{itemize}
        \item If $x$ is on its second hash, this ``first empty slot'' is the $\ell$th slot (so the insertion of $x$ makes $B$ full), and there is another ball in $B$ that is on its first hash, then exchange the positions within $B$ of $x$ and the first ball in $B$ that is on its first hash.
    \end{itemize}
    \item If the first $\ell-1$ slots of $B$ all contain objects on their second hash, place $x$ into the last slot in $B$ (evicting the ball in that slot).
    \item If $x$ was just evicted from its slot under point (2.)~of this list and $x$ has previously been placed into $B$, place $x$ into the last slot in $B$ (evicting the ball in that slot).
    \item Otherwise, there must be balls in the first $\ell-1$ slots of $B$ that are on their first hash. Evict the first one of those balls and place $x$ into its place.
\end{enumerate}

Finally, there is one case where the insertion algorithm may declare failure: If the algorithm performs a sequence of evictions using points (2.)~and (3.), and if the first bin $B$ to be visited in this sequence is visited three times, the algorithm declares failure. (At this point, the algorithm has entered a loop of evictions that will never end.) We will prove, later in our analysis, that every insertion either succeeds or declares failure; and, that the overall probability of the algorithm declaring failure, across all insertions, is $O(1/(\ell^6m))$.

Say that the first hash $h_1(x)$ of a key $x$ is \underline{revealed} when the key is first inserted, and that the second hash $h_2(x)$ is \underline{revealed} the first time that the key is placed in bin $h_2(x)$ (or, equivalently, is evicted from $h_1(x)$). It will be helpful to define the following event $\tau$ which captures the point in time at which the total number of hashes revealed crosses a certain natural threshold:
\begin{definition}
Let $\tau$ be the event that at least $(2-\ell^{-.49})\ell m$ hashes have been revealed.
\end{definition}
Note that $\tau$ may never occur during our insertion process, as it is possible that before this time, our algorithm fails and there is no assignment of balls to bins under our algorithm, and it is also possible that the insertions all complete without $\tau$ occurring.

A nice feature of our algorithm is that, by the time $\tau$ occurs we will likely have reached a \emph{very} high load factor:

\begin{lemma}\label{epsstarfirst}

Let $\epsilon_{\tau}$ be the random variable such that $\epsilon_{\tau}\ell m$ is the number of free slots in the table at the time when the event $\tau$ occurs (set $\epsilon_\tau=\tilde\epsilon$ if $\tau$ never occurs). Then with high probability in $m$, we have that $\epsilon_\tau\le 4e^{\ell^{.51}}\tilde\epsilon$. Here, $\tilde\epsilon$ is as defined in Definition \ref{epstildedef}.
\end{lemma}

\begin{proof}
Every time that we reveal a hash to a bin that has empty slots remaining, we always put that ball in the bin. Therefore, it suffices to prove the following variation of the lemma (which abstracts away the specifics of how the hash table works). Suppose that we reveal $2\ell m$ random hashes (independently and uniformly at random in $[m]$); and that, each time we reveal a hash $B \in [m]$, if bin $B$ still has any free slots, we fill one of them. Let $\epsilon_\tau \ell m$ be the number of remaining free slots after $(2 - \ell^{-.49})\ell m = \ell m - \ell^{.51} m$ reveals, and let $\tilde\epsilon \ell m$ be the number of remaining free slots after all $\ell m$ reveals (matching Definition \ref{epstildedef}). We wish to show with high probability in $m$ that $\epsilon_\tau\le 4e^{\ell^{.51}}\tilde\epsilon$.

The key to the proof is to analyze the number of slots filled by the final $\ell^{.51} m$ reveals. For a slot that is free prior to these reveals, the probability of that slot remaining free during all $\ell^{.51} m$ reveals is at least
$$(1 - 1/m)^{\ell^{.51} m} \ge e^{-\ell^{.51}}/2$$
for sufficiently large $m$. It follows that 
$$\E[\tilde\epsilon] \ge \E[\epsilon_\tau]e^{-\ell^{.51}}/2.$$
Both random variables $\tilde\epsilon \ell m$ and $\epsilon_\tau \ell m$ have the feature that, if we change any one of the hash reveals, it changes the outcome of $\tilde\epsilon \ell m$ and $\epsilon_\tau \ell m$ by at most one. Therefore, by McDiarmid's Inequality, we have with high probability in $m$ that $\tilde\epsilon \ell m$ and $\epsilon_\tau \ell m$ are within $O(\sqrt{\ell m \log m}) \le m^{.51}/\ell$ of their means. Thus, with high probability in $m$, we have
$$\tilde\epsilon \ge \epsilon_\tau \cdot e^{-\ell^{.51}}/2 - m^{-.49}.$$
Recall that, as noted in Subsection \ref{boundingsubsect}, $\mb{E}[\tilde\epsilon]=\Omega(m^{-.47})$, which implies by Lemma \ref{boundingepstilde} that, with high probability in $m$, $\tilde\epsilon=\Omega(m^{-.48})$. Since $\ell = O(\log m)$, the fact that $\tilde\epsilon=\Omega(m^{-.48})$ combines with the centered equation above to imply that
$$\tilde\epsilon \ge \epsilon_\tau \cdot e^{-\ell^{.51}}/4.$$

\end{proof}

Now, we will analyze some useful properties that this algorithm has that will help our analysis.

\begin{definition}
A bin $B$ is \underline{corrupt} if its first $\ell-1$ slots are filled with balls that are on their second hash.
\end{definition}
Note that we only go to point (2.)~if $B$ is corrupt. Also, note that any bin that becomes corrupt stays corrupt, as corrupt bins can then only have balls placed into their final slot.

Our analysis will argue that the insertion algorithm makes many ``fresh probes'' (probing, for some ball $x$, either $h_1(x)$ or $h_2(x)$ for the first time ever). The following lemma tells us that, whenever a \emph{non-fresh} probe performed (a ball is returned to a bin $h_i(x)$ that it has already been in formerly), the bin $h_2(x)$ is necessarily corrupt (and was corrupt on $x$'s first eviction from $h_1(x)$). 
\begin{lemma}\label{corruptreturn}
If a ball $x$ is ever returned to its first bin, $h_1(x)$, after having been evicted from it in the past, then $h_2(x)$ was a corrupt bin at the time of $x$'s initial eviction from $h_1(x)$. Additionally, $x$ is placed into the last slot of $h_1(x)$ when returned to $h_1(x)$.
\end{lemma}
\begin{proof}
Points (1.)~and (4.)~do not evict balls that are on their second hash. Therefore, the only way for a ball $x$ to be returned to $h_1(x)$ after being moved to $h_2(x)$ is through points (2.)~and (3.). Both points (2.)~and (3.)~evict the last slot in a bin. So if, when $x$ is put into $h_2(x)$, the slot in $h_2(x)$ it gets is not the last slot, then $x$ will never again be moved from $h_2(x)$. Thus the only way for $x$ to ever be returned to $h_1(x)$ is if $x$ was put into the last slot of $h_2(x)$.

Now, we claim that the only way for $x$ to be put into the last slot of $h_2(x)$ is if $h_2(x)$ is already corrupt, completing the proof. Point (1.)~will only place $x$ into the final slot of $h_2(x)$ if $h_2(x)$ is already corrupt. Point (2.)~only happens if $h_2(x)$ is corrupt. The first time that $x$ is placed into $h_2(x)$ cannot happen by point (3.). Point (4.)~will not place $x$ into the final slot of $h_2(x)$. So, the only ways that $x$ can be put into the final slot of $h_2(x)$, via points (1.)~and (2.), can only happen if $h_2(x)$ is already corrupt.

Finally, any time $x$ is evicted from $h_2(x)$ must come by point (2.), putting $x$ back into the last slot in $h_1(x)$ by points (2.)~or (3.)~(also noting that $h_1(x)$ can no longer have an empty slot if $x$ was evicted from it).
\end{proof}
Lemma \ref{corruptreturn} quickly gives us a corollary:
\begin{lemma}\label{uniformpurefirst}
Any ball $x$ evicted from $h_1(x)$ by point (4.)~has never previously been placed in $h_2(x)$.
\end{lemma}
\begin{proof}
This follows as the contrapositive of Lemma \ref{corruptreturn}, as any ball $x$ that is returned to $h_1(x)$ will be in the last position of $h_1(x)$, and therefore will not be in the first $\ell-1$ positions that (4.)~evicts from. 
\end{proof}

 
\subsection{The Corruption Graph}
We will consider an auxiliary graph, which we call the \underline{corruption graph} $G$ (technically, a multigraph, as it could have multiedges or loops). This will be a subgraph of the (essentially) Erd\H os--R\'enyi graph $\{(h_1(u), h_2(u)) \mid u \text{ is an element}\}$. In particular, define the corruption graph as the subgraph consisting of all balls (edges) whose second hash goes to a corrupt bin -- excluding the $\ell-1$ second hashes that are the first to land on that corrupt bin during the insertion process.

Note that the corruption graph $G$ develops over time (as bins become corrupt). When analyzing a given insertion, we will often refer to $G$ at a specific point in time (i.e., including only bin-corruption events that occur up to that point). Lemma \ref{epsstarfirst} tells us that the $\epsilon$ that we care about for Theorem \ref{thmfirst} are likely to occur before $\tau$ occurs, so we can focus on properties that the corruption graph is likely to have at any point before $\tau$ occurs.

One reason that the corruption graph $G$ is important to us is that, if an insertion declares failure, then it turns out that we can ``blame $G$'', showing that $G$ contained a bicyclic component. Note that the corruption graph may have multiedges or loops, so a bicyclic component is defined as a set of $k$ vertices such that the induced subgraph on those $k$ vertices has at least $k+1$ edges.
\begin{lemma}\label{bicycliconlybarrier}
Every insertion either succeeds or declares failure. Moreover, if an insertion declares failure, then the corruption graph (at the point in time where failure is declared) contains a bicyclic component. Moreover, no insertion (even one that fails) will ever perform a sequence of (2.)~and (3.)~evictions in which it visits the same slot more than 3 times.
\label{lem:fail}
\end{lemma}
\begin{proof}
During the insertion algorithm, point (1.)~can happen at most once, and point (4.)~can happen at most $m \ell$ times, since each instance of (4.)~reveals the second hash for some ball for the first time. Therefore, the only way the insertion algorithm can loop indefinitely is if it alternates between (2.)~and (3.)~indefinitely. 

However, any sequence of calls to (2.)~and (3.)~can be viewed as performing a \emph{standard} cuckoo-hashing eviction chain (with no bins) on the elements $u$ that have edges in the corruption graph. By the standard analysis of cuckoo hashing (see, e.g., \cite{PR01} or the discussion in \cite{KM25}), we can conclude that the only way for such an eviction chain to go indefinitely is if it takes place in a component of the corruption graph that is bicyclic; that, in any sequence of evictions that alternates points (2.)~and (3.), if the sequence goes on indefinitely, then the first vertex to be visited three times will be the one at the start of the sequence (at which point the insertion algorithm will declare failure); and that in any sequence of evictions that alternates points (2.)~and (3.), if the sequence \emph{does not} go on indefinitely, then each bin is visited at most twice (so the algorithm does not declare failure). These three observations together imply the lemma. 
\end{proof}

In addition to determining whether our algorithm fails, the structure of $G$ also determines whether our algorithm is efficient, as captured by the following lemma. Roughly speaking, what the lemma says is that, so long as the connected components of $G$ are small (each roughly $O(1)$ size), then our insertion algorithm is fast. This is because, whenever the algorithm is \emph{not} exploring a connected component of $G$, it is instead exploring some hash that has never before been examined (and that therefore has good probability of taking us to a free slot). 
\begin{lemma}\label{runtime}
Let $G$ be the corruption graph immediately prior to the $k$-th insertion (or if one of the first $k$ insertions fail, let it be the corruption graph then), but with the edge corresponding to the $k$-th insertion removed (if present). Suppose we have already proven that:
\begin{enumerate}
    \item For a uniformly random bin $B$, the expected length of the longest path containing $B$ (and using each edge at most once) in $G$ is $\le C$ (here $G$ is a random variable). Here, path length is measured in number of edges. 
    \item Immediately prior to the insertion, we have with high probability in $m$ that at least an $\epsilon$ fraction of bins contain at least one free slot. 
    \item  Immediately prior to the insertion, there are at most $m/2$ corrupt bins, with high probability in $m$.
    \item In the corruption graph $\overline{G}$ after the insertion, the  longest path (anywhere in the graph) has length at most $\polylog m$ with high probability in $m$. 
    \end{enumerate}
Then, the expected number of evictions to complete the $i$-th insertion (where the number is $0$ if some previous insertion failed, and is the number of evictions to failure if the $i$-th insertion fails) is $\epsilon^{-1} + O(1 + C \epsilon^{-1})$.
\end{lemma}
\begin{proof}

Let $x_0$ be the ball being inserted. Let $G$ be the corruption graph before the insertion occurs, but with the edge corresponding to $x_0$ removed if it is present. Let $\overline{G}$ be the corruption graph after the insertion. Define the event $E$ to be the event that we ever evict a ball $u \neq x_0$ whose edge is in $\overline{G}$ but not in $G$. We will break the number of evictions $I$ into two pieces, $I_1 = I \cdot \mathbf{1}(\overline{E})$ and $I_2 = I \cdot \mathbf{1}(E)$. Most of the proof is spent bounding $\E[I_1]$.

Throughout the analysis of $I_1$, we will fix the graph $G$, and define $C(G)$ to be the expected length of the longest path in the component containing a uniformly random vertex of $G$. We will also condition on the high-probability event that there are at least $\epsilon m$ bins with free slots, that there are at most $m/2$ corrupt bins in the graph $G$, and that every path in every component of $G$ has length at most $\polylog m$. To show that $\E[I_1] = \epsilon^{-1} + O(1 + C \epsilon^{-1})$ (without these conditions), it suffices to show that, with these conditions in place, the expected number of  evictions is $\epsilon^{-1} + O(1 + C(G) \epsilon^{-1})$.  

We can think of the insertion as proceeding as follows: First place the ball $x_0$ being inserted in bin $h_1(x_0)$ (we will refer to this as \underline{eviction} $0$). Then, perform a sequence of evictions, where \underline{eviction $i$} moves some ball $x_i$ between bins. Call an eviction \underline{special} if it evicts a ball $x_i$ whose second hash has never yet been revealed (i.e., ball $x_i$ has never used its second hash before), and whose edge is not in the graph $G$. 

Assuming that $E$ does not occur, then the following is true: For any maximal sequence of non-special evictions, these evictions all take place in a single component of $G \cup \{(h_1(x_0), h_2(x_0))\}$ (and involve balls that have edges in $G\cup \{(h_1(x_0), h_2(x_0))\}$). By Lemma \ref{lem:fail}, if such a sequence starts at some bin $b$, then the number of evictions in the sequence is at most $O(P_b)$, where $P_b$ is the length of the longest path in $G \cup \{(h_1(x_0), h_2(x_0)\}$ containing node $b$. (If node $b$ has no incident edges in $G \cup \{(h_1(x_0), h_2(x_0)\}$, then $P_b = 0$.) If, in total, there are $J$ special evictions performed, define $b_0, b_1, b_2, \ldots, b_J \in [m]$, where $b_0$ is $h_1(x_0)$, and where, for $0 < j \le J$, $b_j$ is defined to be $h_2(x)$ for whichever ball $x$ is evicted by the $j$-th special eviction. Then, 
$$\E[I_1] \le \sum_{j \ge 0} \Pr[b_j\text{ exists}] \cdot (1 + O(\E[P_{b_j} \mid b_j \text{ exists}])).$$
Define $P_b'$ to be the length of the longest path in the component of $G$ (rather than $G \cup (h_1(x_0), h_2(x_9))$) containing $b$. Let $Y_j$ be the indicator random variable for the event that $b_j$ exists (that is, we perform at least $j$ special evictions) \emph{and} that  $b_j$ is in the same connected component in $G$ as one of $h_1(x_0)$ or $h_2(x_0)$. Then,  
\begin{equation}\E[I_1] \le \E\left[O(P_{h_1(x_0)}' + P_{h_2(x_0)}' + 1) \cdot \sum_{j \ge 0} Y_j\right] + \sum_{j \ge 1} \Pr[b_j\text{ exists}] \cdot (1 + O(\E[P_{b_j}' \mid b_j \text{ exists}])).
\label{eq:I1twoparts}
\end{equation}
Recall that each component of $G$ has longest-path-length at most $\polylog m$. For $j \ge 1$, if the $j$-th special eviction evicts some ball $y \neq x_0$, then the eviction has probability at most $\polylog m / m$ of satisfying $Y_j = 1$. Since $P_{h_1(x_0)} + P_{h_2(x_0)} \le \polylog m$, and since $x_0$ can be subject to at most two special evictions (eviction $j =0$, and possibly some eviction $j \ge 1$), it follows that
\begin{align*}
    & \E\left[(P_{h_1(x_0)}' + P_{h_2(x_0)}' + 1) \cdot \sum_{j \ge 0} Y_j\right] \\
    & \le 2 \E[P_{h_1(x_0)}' + P_{h_2(x_0)}' + 1] + \sum_{j \ge 1} \Pr[b_j \text{ exists}] \cdot \E[P_{h_1(x_0)}' + P_{h_2(x_0)}' + 1 \mid b_j \text{ exists}] \cdot \frac{\polylog m}{m} \\
    & \le O(C(G) + 1) + \sum_{j \ge 1} \Pr[b_j \text{ exists}] \cdot \E[P_{h_1(x_0)}' + P_{h_2(x_0)}' + 1 \mid b_j \text{ exists}] \cdot \frac{\polylog m}{m} \\
    & \le O(C(G) + 1) + \sum_{j \ge 1} \Pr[b_j \text{ exists}] \cdot \frac{\polylog m}{m},
\end{align*}
where the final step uses the fact that every component in $G$ has size at most $\polylog m$. 

Note that $\Pr[b_j \text{ exists}]$ is at most the probability that the insertion reveals at least $j + 1$ fresh hashes (hashes never yet revealed before) without completing, which is at most $(1 - \epsilon)^{j + 1}$.  Therefore,
\begin{equation}E\left[(P_{h_1(x_0)}' + P_{h_2(x_0)}' + 1) \cdot \sum_{j \ge 0} Y_j\right] \le O(C(G) + 1) + \sum_{j \ge 1} (1 - \epsilon)^{j + 1} \cdot \frac{\polylog m}{m} \le O(C(G) + 1).
\label{eq:I1part1}
\end{equation}
Now, turning our attention to the second sum in \eqref{eq:I1twoparts}, observe that
$$\sum_{j \ge 1} \Pr[b_j\text{ exists}] \cdot (1 + O(\E[P_{b_j}' \mid b_j \text{ exists}])) \le \sum_{j \ge 1}(1 - \epsilon)^{j + 1} \cdot (1 + O(\E[P_{b_j}' \mid b_j \text{ exists}])).$$
For $j \ge 1$, the distribution of $b_j$ (conditioned on it existing) is uniformly random across all non-corrupt bins in $G$. Since there are $\Theta(m)$ such bins, it follows that 
\begin{align*} 
& \E[ P_{b_j}' \mid b_j \text{ exists}] \\ 
& \le \sum_{\text{non-corrupt }b} P_b'\cdot O(1/m) \le O\left(\sum_{\text{all bins }b} P_b'/m\right) \\
& = O(C(G)).
\end{align*}
Therefore, $\E[P_{b_j}' \mid b_j \text{ exists}] = O(C(G))$, and 
\begin{equation}\sum_{j \ge 1} \Pr[b_j\text{ exists}] \cdot (1 + O(\E[P_{b_j}' \mid b_j \text{ exists}])) \le
\sum_{j > 1} (1 - \epsilon)^{j + 1} \cdot (1 + O(C(G))) \le \epsilon^{-1} + O(C(G) \epsilon^{-1}).
\label{eq:I1part2}
\end{equation}
Combining \eqref{eq:I1twoparts}, \eqref{eq:I1part1}, and \eqref{eq:I1part2}, we can conclude that
$$\E[I_1] \le O(1 + C(G)) + \epsilon^{-1}.$$

Finally, we must also bound $\E[I_2]$ (here we do not condition on $G$). Let $P$ be the length of the longest path in \emph{any} component of $\overline{G}$. By assumption, we have $P \le \polylog m$ with high probability in $m$. Any sequence of non-special evictions has length at most $O(P) \le \polylog m$
Moreover, the number of special evictions is, with high probability in $m$ at most $O(\epsilon^{-1} \log m)$ (since $\Pr[b_j \text{ exists}] \le (1 -\epsilon)^j + 1/\poly(m)$). Therefore, the total number of evictions is, with high probability in $m$, at most some number
$$q := \polylog m + O(\epsilon^{-1} \log m).$$
For $E$ to occur, there must be \emph{two different balls} $x_i$ and $x_j$ that are evicted during the insertion, whose second hashes were never revealed beforehand, and whose second hashes happen to be equal. (In particular, it takes one ball $x_i$ to make a new bin $h_2(x_i)$ corrupt, and then a second ball $x_j$ with $h_2(x_j) = h_2(x_i)$ to make event $E$ occur.) The expected number of such pairs of balls is at most $1/\poly(m) + O(q^2/m) \le \polylog m / m$, meaning that $\Pr[E] \le \polylog m / m$. The expected value of $I_2 = I \cdot \mathbf{1}(E)$ therefore satisfies
$$\E[I_2] \le q \cdot \Pr[E] + 1/\poly(m) = (\epsilon^{-1} \polylog m) / m   = o(1)$$
as desired.
\end{proof}

The following lemma gives us a simple stochastic process that produces a graph $G'$ which dominates the corruption graph $G$:

\begin{lemma}
Consider a threshold $T \le m / 20$, and let $G$ be the corruption graph at the final point in time during which it has $\le T$ corrupt vertices, or when $\tau_1$ occurs (whichever comes first). Then, there exists a graph $G'$ that is constructed via the following process, and where the edges (resp.~corrupt vertices) in $G$ are a subset of the edges (resp.~corrupt vertices) in $G'$. 

The graph $G'$ is constructed by selecting a uniformly random set of $T$ (distinct) vertices to be corrupt, and then adding (unordered) edges to each corrupt vertex $B$ as follows:  
\begin{enumerate}
    \item For each bin $B'$, $\mathrm{Bin}(\ell,\frac{1.1}{m})$ edges are added from $B$ to $B'$. This process occurs independently for different $B'$.
    \item $\mathrm{Bin}(\ell m,\frac{1.1}{m})$ edges of the form $(B, B')$ are added, where each $B'$ is independent and uniformly random.
\end{enumerate}
The two steps above are independent both of each other, and across all corrupt vertices.
\label{lem:Gprime}
\end{lemma}
\begin{proof}
We will consider how $G$ evolves over time (as elements are inserted), and construct $G'$ alongside it. As we perform the construction, we will also keep track for each non-corrupt bin $b$ of the quantity defined to be: the total number of balls $u$ inserted so far with $h_2(u) = b$, and where ball $u$ is in bin $b$. We will define the \underline{system state} at any given moment to be the tuple $S = (G, G', \{y_b\})$ encoding all of $G$, $G'$, and the $y_b$ values for each non-corrupt bin $b$.

Consider the moment in which a bin $B$ becomes corrupt in $G$ (and assume $B$ is among the first $T$ bins to become corrupt). We will also declare bin $B$ corrupt in $G'$, and we will argue below that we can add edges to $G'$ in such a way that (1) all edges added to $G$ get added to $G'$, (2) the edges added to $G'$ are as described in the lemma statement, and (3) which edges are added to $G'$ is fully independent of the system state prior to $B$ becoming corrupt. 

There are two types of edges that get added to $G$ as a result of the bin becoming corrupt: Edges $(h_1(x), h_2(x))$ for balls $x$ that are already present, that have never-yet used their second hash, and that satisfy $h_2(x) = B$ (call these Type 1 edges); and edges $(h_1(x), h_2(x))$ for balls that have not yet been inserted, but whose second hash is $h_2(x) = B$ (call these Type 2 edges). 

We can think of Type 1 edges as being added to $G$ by the following process. For each bin $B'$, and for each ball $x$ in $B'$ (with $h_1(x) = B'$) that has never-yet used its second hash, check if $x$'s second hash is $B$ -- if so, add the edge $(h_1(x), h_2(x)) = (B', B)$ to $G$. If the number of non-corrupt vertices prior to $B$ becoming corrupt was $K$, then each ball $x$ in $B'$ that has not yet used its second hash (and that does not already correspond to an edge in $G$) independently has probability $1/K$ of satisfying $h_2(x) = B$. Since there are at most $\ell$ such balls $x$, the number of edges added from $B'$ to $B$ is at most a binomial random variable $\mathrm{Bin}(\ell,1/K)$. By assumption $K \ge m - T \ge .95m$, so the number of edges added from $B'$ to $B$ is dominated by the random variable $\mathrm{Bin}(\ell, 1.1/m)$. This is true independently for each bin $B'$. Therefore, it is possible to add edges to $G'$ so that $G'$ receives all the Type 1 edges that $G$ receives, and so that the number of edges $G'$ receives from each bin $B'$ is an independent $\mathrm{Bin}(\ell, 1.1/m)$ random variable (which is also fully independent of the system state prior to $B$ becoming corrupt). 

Likewise, we can think of Type 2 edges as being added to $G$ by the following process. Each of the at-most $\ell m$ keys $x$ that remain to be inserted (and that do not yet correspond to edges in $G$) have independent uniformly random $h_1(x)$ value and have probability $1/K \le 1.1 / m$ of satisfying $h_2(x) = B$. This means that the number of Type 2 edges added to $G$ is dominated by $\mathrm{Bin}(\ell m, 1.1/m)$. Therefore, it is possible to add edges to $G'$ so that $G'$ receives all the Type 2 edges that $G$ receives, so that the number of edges added to $G'$ is a $\mathrm{Bin}(\ell m, 1.1/m)$ random variable, so that each edge added is independent and uniformly random from $\{(B', B) \mid B' \in [m]\}$, and so that the edge additions to $G'$ are fully independent of the system state prior to $B$ becoming corrupt.

The above process tells us how to add edges to $G'$ whenever a bin $B$ becomes corrupt in $G$. In addition, if at the end of the construction of $G$ (the final moment at which $G$ has $\le T$ corrupt vertices) $G$ has fewer than $T$ corrupt vertices, we add additional corrupt vertices to $G'$ until the total number of corrupt vertices is $T$, and we do it as follows: We continue to increment random $y_b$ values (we call these \underline{fake increments}), making a bin $b$ corrupt once its $y_b$ hits $\ell - 1$ (and stopping once the number of corrupt vertices is $T$). Each time a bin $b$ becomes corrupt, we add edges as described in the lemma statement.

The resulting graph $G'$ has exactly $T$ corrupt vertices, and has the property that, when a corrupt vertex is introduced, it gets edges according to the distribution described in the lemma statement. The final (and most subtle) step in the proof is to argue that the actual set of corrupt vertices in $G'$ is uniformly random (and independent of the edges). For this, we must show that, when a bin $B$ becomes corrupt in $G'$, it is selected \emph{uniformly random out of all remaining uncorrupted bins}, and \emph{independently} of the current system state.

This part of the argument is where the $\{y_b\}$ values finally come into play. Notice that a bin $B$ becomes corrupt in $G'$ exactly when $y_B$ gets incremented to $\ell - 1$. We will argue that, each time some $y_b$ gets incremented, the bin $b$ is uniformly random (out of the non-corrupt bins) and independent of the system state -- it follows that the $K$ bins that become corrupt are uniformly random and independent of the edges in the graph.

By construction, fake increments to $y_b$ (increments that happen at the end of the construction bring the number of corrupt nodes in $G'$ to $K$) are each to uniformly random bins $b$ (out of those not-yet-corrupt). Thus we focus in the rest of the proof on non-fake increments to $y_b$. 

Besides a fake increment, the only way for a counter $y_b$ to change (for a non-corrupt bin $b$) is for some ball $x$ with $h_2(x) = b$ to be evicted from its first bin. Note that, when an element $x$ is evicted for the first time, if $x$ does not already correspond to an edge in $G$, then we are guaranteed that its \emph{second hash is uniformly random out of all non-corrupt bins}. So the effect of evicting $x$ will be exactly equivalent to incrementing $y_b$ for a uniformly random non-corrupt bin $b$ (independent of the current system state). 

Thus, each time some $y_b$ gets incremented, the bin $b$ is uniformly random and independent of the system state.

Putting the pieces together, we have argued that the construction of $G'$ is equivalent to the following: We increment counters $y_B$ for uniformly random not-yet-corrupt bins $B$ (independent of the current system state), and declare a bin $B$ to be corrupt once its counter reaches $\ell - 1$, stopping once we have $K$ corrupt bins; when a bin becomes corrupt, we add edges as described in the lemma statement, and in a way that is also independent of the current system state. This, overall, is equivalent to picking $K$ random vertices and adding edges as described in the lemma statement. 

\end{proof}

We will often find ourselves in situations where it suffices to reason about $G'$ rather than $G$. For these situations, the following lemma about $G'$ will be critical:

\begin{lemma}\label{probedgeexists}
Consider the state of the dominating corruption graph $G'$, described in Lemma \ref{lem:Gprime}, at any given moment. Any edge (including any loop) has probability at most $\frac{5\ell}{m}$ of existing in the graph, and this remains true when conditioning on which vertices are corrupt. Similarly, the probability of a multi-edge existing with multiplicity at least $c$ is at most $\left(\frac{5\ell}{m}\right)^c$, again regardless of which vertices are corrupt; and the probability of any $c'$ edges (possibly with repeats) all existing is at most $\left(\frac{5\ell}{m}\right)^{c'}$, again regardless of which vertices are corrupt.
\end{lemma}
\begin{proof}
Take an edge $(x,y)$ (where we could have $x=y$). This edge could appear in $G'$ either when $x$ becomes corrupt or when $y$ becomes corrupt. We first show that for any $c_x\in\mb{N}$, we have that the probability of it being added with multiplicity exactly $c_x$ when $x$ becomes corrupt is at most $\left(\frac{2.2\ell}{m}\right)^{c_x}$ (and, by symmetry, the same holds for $y$ in place of $x$).

Consider some $c_1, c_2 \ge 0$ with $c_1 + c_2 = c_x$. When $x$ becomes corrupt, $G'$ receives $\mathrm{Bin}\left(\ell,\frac{1.1}{m}\right)$ edges from each bin $z$ to bin $x$ (call these Type 1 edges); and then additionally receives $\mathrm{Bin}\left(\ell m,\frac{1.1}{m}\right)$ edges from $x$ to random other bins $z$. Thus \[\mb{P}\left(\mathrm{Bin}\left(\ell,\frac{1.1}{m}\right)=c_1\right)\le\left(\frac{1.1\ell}{m}\right)^{c_1}\]is a bound on the probability that $c_1$ Type 1 edges are added between $x$ and $y$, and  \[\mb{P}\left(\mathrm{Bin}\left(\ell m,\frac{1.1}{m^2}\right)=c_2\right)\le\left(\frac{1.1\ell}{m}\right)^{c_2}\]is a bound on the probability of $c_2$ Type 2 edges being added between $x$ and $y$ (and these two probabilities are independent). Then, summing over the options for $c_1, c_2$, \[\sum_{c_1=0}^{c_x}\left(\frac{1.1\ell}{m}\right)^{c_1}\left(\frac{1.1\ell}{m}\right)^{c_x-c_1}=(c_x+1)\left(\frac{1.1\ell}{m}\right)^{c_x}\le\left(\frac{2.2\ell}{m}\right)^{c_x}\]is an upper bound on the probability that $c_x$ total edges are added between $x$ and $y$ when $x$ becomes corrupt. 

Finally, for the edge $(x, y)$ to appear with multiplicity $c$ overall, it needs to appear with multiplicity $c_x$ when $x$ becomes corrupt and $c_y$ when $y$ becomes corrupt, where $c_x+c_y=c$. Using that the edges that appear at $x$'s corruption are independent of the edges that appear at $y$'s corruption, the probability of it appearing with multiplicity exactly $c$ overall is at most \[
\sum_{c_x=0}^c\left(\frac{2.2\ell}{m}\right)^{c_x}\left(\frac{2.2\ell}{m}\right)^{c-c_x}=(c+1)\left(\frac{2.2\ell}{m}\right)^c.
\]Then we have that the probability of it appearing with multiplicity at least $c$ is\[\sum_{a=c}^\infty (a+1)\left(\frac{2.2\ell}{m}\right)^a\le\left(\frac{5\ell}{m}\right)^c\](using $\ell=O(\log m)$)
as desired.

We claim that the edges generated by (1.) in Lemma \ref{lem:Gprime} are independent, while those generated by (2.) can only have a negative dependency between the existence of one edge (with at least a certain multiplicity) and the existence of a different edge (with at least a certain multiplicity). To justify the latter statement, note process (2.) is equivalent to the following: for a given $B$ and each of $\ell m$ independent trials, we have for each bin $B'$ a \emph{disjoint} $\frac{1.1}{m^2}$ probability that the trial will add $(B,B')$, and the remaining $1-\frac{1.1}{m}$ probability it will add no edge. Therefore, we see that the probability of an edge existing at least $c$ times is either independent of another edge existing at least $c'$ times (if they are not incident), or if they are incident (but not the same edge) one edge existing at least $c$ times makes it conditionally (very slightly) less likely for another edge to exist at least $c'$ times. This gives us the final statement that the probability of any $C$ edges (possibly with repeats) all existing is at most $\left(\frac{5\ell}{m}\right)^C$, again independently of which vertices are corrupt.
\end{proof}

Thus, to complete the analysis of the algorithm, the key is to prove that $G$ itself behaves nicely -- namely that $G$ most likely does not contain any bicyclic components, and that $G$'s connected components contain only short paths (in expectation). 

\subsection{``Subcritical'' Properties of the Corruption Graph}

The next lemma says that, up until $\tau$ occurs, we have very few total corrupted bins. This, combined with Lemma \ref{probedgeexists}, will be enough for us to prove essentially everything we want to show about the corruption graph $G$.

\begin{lemma}\label{numcorruptbins}
With high probability in $m$ over the insertion process, there will be no time at which $\tau$ has not occurred but the number of corrupted bins is at least $\ell^{-10}m$.
\end{lemma}
\begin{proof}
A bin $B$ is corrupt if and only if there are at least $\ell-1$ balls whose second hashes have been revealed to be $B$.

If $\tau$ has not yet occurred, then there have been at most $((2-\ell^{-.49})\ell m)/2=(1-.5\ell^{-.49})\ell m$ second hashes revealed (as each ball's first hash is revealed before its second hash). It therefore suffices to prove the following: that, if we reveal $K = (1-.5\ell^{-.49})\ell m$ second hashes $R_1, R_2, \ldots, R_K \in [m]$, then the number of bins $B \in [m]$ that are hit at least $\ell - 1$ times is at most $\ell^{-10} m$ with high probability.

The number of second hashes revealed to be a given bin $B$ is  $\mathrm{Bin}((1-.5\ell^{-.49})\ell m,\frac 1m)$. By Lemma \ref{poisson} with $c=1-.5\ell^{-.49}$,
\[\mb{P}(\mathrm{Bin}((1-.5\ell^{-.49})\ell m,1/m)\ge\ell-1)\le e^{-\ell((.5\ell^{-.49})^2/2}/(1-.5\ell^{-.49})\le e^{-\ell^{.02}/10}\le\ell^{-11}\]for $\ell$ at least a sufficiently large constant.

Now, let $X$ be the number of bins that are hit at least $\ell - 1$ times by $R_1, R_2, \ldots, R_K$. The above shows that $\mb{E}[X]<\ell^{-11}m$. Since $X$ is determined by $K = O(m\ell)$ independent random variables $R_1, R_2, \ldots, R_K$, and since each $R_i$ can change $X$ by at most 1, we can apply McDiarmid's Bounded Difference Inequality to say that \[
\mb{P}(X\ge \ell^{-10}m)\le\mb{P}(X\ge 2\ell^{-11}m)\le\mb{P}(X-\mb{E}[X]\ge\ell^{-11}m)\le e^{-\frac{(\ell^{-11}m)^2}{\ell m}}\le e^{-\frac{m}{\ell^{23}}}.
\]
We note that $\mb{P}(X<\ell^{-10}m)\ge 1-e^{-\frac{m}{\ell^{23}}}$ is then with high probability in $m$ as $\ell=O(\log m)$.
\end{proof}
Intuitively, the bound in Lemma \ref{numcorruptbins} is small enough that the corruption graph has similar properties to the subcritical regime of the Erd\H os--Renyi random graph (that is, $G(n,p)$ for $p\le\frac{1-\Omega(1)}{n}$ or $G(n,m)$ for $m\le (\frac 12-\Omega(1))n$) \cite{ErdosRenyi,FriezeKaronski}. In particular, we can now argue that the corruption graph most likely does not contain any bicyclic component prior to $\tau$ occurring (Lemma \ref{nobicyclic}) and that the components tend to have only short paths (Lemma \ref{CorruptComponentSize} and Lemma \ref{longestpathpolylog}). 

\begin{lemma}\label{nobicyclic}
With probability at least $1-O(\ell^{-6}m^{-1})$, the point in the insertion process where $\tau$ is reached comes before any point in the insertion process where we fail to insert any additional elements.
\end{lemma}
\begin{proof}
For failure to occur, there has to be a bicyclic component in the corruption graph ($k+1$ edges on $k$ vertices), as explained in Lemma \ref{bicycliconlybarrier}. As this is an increasing graph property and the corruption graph is stochastically dominated by the dominating corruption graph $G'$ described in Lemma \ref{lem:Gprime}, it suffices to prove that there is no bicyclic component in $G'$ with probability at least $1-O(\ell^{-6}m^{-1})$.

We will count the number of minimal bicyclic components, that is, a vertex set $V$ induces a bicyclic component but no subset $V'$ of $V$ does. If a bicyclic component exists, a minimal bicyclic component must also exist. In a minimal bicyclic component, every $v\in V$ has degree at least two, or else we could remove $v$ and the (at most) one incident edge and maintain a bicyclic component.

Note that in the corruption graph, every edge is incident to at least one corrupt bin. Assume a minimal bicyclic component on $k$ vertices $C$ has $j$ corrupt vertices and thus $k-j$ non-corrupt vertices. If $j\ne k$, let $w\in C$ be a non-corrupt vertex. $C\setminus\{w\}$ is not bicyclic, so it has at most $k-1$ edges. Each of the $k-j-1$ non-corrupt vertices in $C\setminus\{w\}$ has degree at least two (as none connected to $w$). Those $k-j-1$ vertices form an independent set, so $C\setminus\{w\}$ has at least $2(k-j-1)$ edges. This gives us that $2(k-j-1)\le E[C]\le k-1$, so $j\ge\frac{k-1}{2}$.

Let $N_k$ be the number of minimal bicyclic components in the corruption graph with $k$ vertices. We claim that \[
\mb{E}[N_k]\le\binom{\ell^{-10}m}{\lceil(k-1)/2\rceil}\binom{m}{\lfloor(k+1)/2\rfloor}k^{k-2}k^4\left(\frac{5\ell}{m}\right)^{k+1}
\]In the inequality above, the first two factors choose $\lceil(k-1)/2\rceil$ corrupt vertices for the component and then the remaining $\lfloor(k+1)/2\rfloor$ vertices to get $k$ total. Then, any bicyclic component must consist of a spanning tree plus two edges, so we choose one of the $k^{k-2}$ spanning trees on those $k$ vertices followed by at most $k^4$ ways to choose the final two edges. Then, each of the $k+1$ edges (possibly repeated) that we have chosen to exist will be appear in $G'$ with probability at most $(5\ell/m)^{k+1}$ by Lemma \ref{probedgeexists}.

Then for all $k\ge 1$ and $\ell\ge 2$, we get
\begin{align*}
\mb{E}[N_k]&\le\binom{\ell^{-10}m}{\lceil(k-1)/2\rceil}\binom{m}{\lfloor(k+1)/2\rfloor}k^{k-2}k^4\left(\frac{5\ell}{m}\right)^{k+1}
\\&\le m^{\lfloor(k+1)/2\rfloor}(\ell^{-10}m)^{\lceil(k-1)/2\rceil}\frac{k^k}{\lfloor(k+1)/2\rfloor!\lceil(k-1)/2\rceil!}k^2\left(\frac{5\ell}{m}\right)^{k+1}
\\&\le m^{\lfloor(k+1)/2\rfloor}(\ell^{-10}m)^{\lceil(k-1)/2\rceil}\frac{k^k}{((k/(2e))^{k/2})^2}k^2\left(\frac{5\ell}{m}\right)^{k+1}
\\&\le m^{\lfloor(k+1)/2\rfloor}(\ell^{-10}m)^{\lceil(k-1)/2\rceil}(2e)^kk^2\left(\frac{5\ell}{m}\right)^{k+1}
\\&\le m^{-1}\ell^{-5k+5}(2e)^kk^25^{k+1}\ell^{k+1}
\\&\le m^{-1}\ell^{-4k+6}k^250^k
\end{align*}
While the above works for all $k\ge 1$, we can achieve a better dependency on $\ell$ in the $k=1$ and $k=2$ cases. For $k=1$, we note that the one vertex involved must be corrupt (a corrupt bin with two loops), so we get
\[\mb{E}[N_1]\le(\ell^{-10}m)\left(\frac{5\ell}{m}\right)^2=5m^{-1}\ell^{-8}.\]
For $k=2$ we have \begin{align*}
\mb{E}[N_2]&\le m^{\lfloor(k+1)/2\rfloor}(\ell^{-10}m)^{\lceil(k-1)/2\rceil}(2e)^kk^2\left(\frac{5\ell}{m}\right)^{k+1}=m^2\ell^{-10}(4e)^2\left(\frac{5\ell}{m}\right)^3
\\&\le 10^6m^{-1}\ell^{-7}.
\end{align*}

Then letting $N$ be the total number of bicyclic components, we have\begin{align*}
\mb{E}[N]&\le 5m^{-1}\ell^{-8}+10^6m^{-1}\ell^{-7}+\sum_{k=3}^{\infty}m^{-1}\ell^{-4k+6}k^222^k\\&\le m^{-1}\ell^{-6}\left(5\ell^{-2}+10^6\ell^{-1}+50^3\sum_{k=0}^\infty\ell^{-6k}(k+3)^250^k\right)
\\&\le 10^8m^{-1}\ell^{-6}\say{for all $\ell\ge 2$}
\end{align*}
Therefore, Markov's inequality tells us that with probability at least $1-\frac{10^8}{\ell^6m}$, the corruption graph has no bicyclic component.
\end{proof}

Note that throughout this Section \ref{firstalgsection}, the probability of a bicyclic component existing in the corruption graph, which is at most $\frac{10^8}{\ell^6m}$, will be our only failure mode that does not give a ``with high probability in $m$'' statement. In other words, every other assumption we make holds with probability $O(m^{-10})$, say.

Now, in order to invoke Lemma \ref{runtime}, we need to bound the expected longest-path length starting from a random vertex in the corruption graph.

\begin{lemma}\label{CorruptComponentSize}
Consider the corruption graph $G$ when $\tau$ occurs (or if the insertion process fails before $\tau$, let $G$ be the corruption graph before the failed insertion). Choose a bin $B$ uniformly at random and let $W(B)$ be the longest path in $G$ that contains $B$. Then \[
\mb{E}[W(B)]\le 200\ell^{-4}.
\]
\end{lemma}
Again, we count path length by the number of edges. Also, we do allow paths to repeat vertices, but we do not allow them to repeat edges (a different edge connecting the same two vertices, in case of multi-edges, is allowed).
\begin{proof}
As in Lemma \ref{bicycliconlybarrier}, it suffices to prove this in the dominating corruption graph $G'$ from Lemma \ref{lem:Gprime}, as this $W(B)$ is an increasing graph property (cannot decrease as new edges are added). We can also assume the result of Lemma \ref{numcorruptbins} holds, as this happens with high probability in $m$. In the case when Lemma \ref{numcorruptbins} (or any ``with high probability in $m$'' statement) fails, we can upper bound $W(B)$ by $\ell m$ (as the corruption graph can never have more than $\ell m$ total edges), adding an $o(\ell^{-4})$ factor to the overall $\mb{E}[W(B)]$ (for sufficiently large $m$, using that $\ell=O(\log m)$).

Let $P_k$ be the number of vertices on a path in the corruption graph of length $k$. Note that any path must have that at least $k/2$ of its $k+1$ vertices are corrupt, as there are no edges between non-corrupt vertices. Then \begin{align*}
\mb{E}[P_k]&\le (k+1)2^{k+1}(\ell^{-10}m)^{\lceil k/2\rceil}m^{\lfloor k/2+1\rfloor}\left(\frac{5\ell}{m}\right)^k\\&\le 2m(k+1)\ell^{-4k}10^k,
\end{align*}
where the first inequality uses the following reasoning: to choose a path, we choose which of the $k+1$ vertices in it will be corrupt in $\le 2^{k+1}$ ways, then choose the $\ge k/2$ corrupt vertices in order and the $\le k/2$ remaining vertices in order, then the $k-1$ edges required with multiplicities have a $\left(\frac{5\ell}{m}\right)^{k-1}$ chance of all appearing by Lemma \ref{probedgeexists}.

Therefore, if we choose a vertex uniformly at random from the $m$ vertices, the probability that the longest path containing it has length at least $k$ is upper bounded by the probability that it is on a path of length $k$, which is at most $2(k+1)\ell^{-4k}10^k$.

Then as we have $\mb{P}(W(B)\ge k)\le 2(k+1)\ell^{-4k}10^k$ for all $k\ge 2$, we see \[\mb{E}_{\text{random bin $B$, random hashes}}[W(B)]\le\sum_{k=1}^\infty 2(k+1)\ell^{-4k}10^k\le\ell^{-4}\sum_{k=1}^\infty 2(k+1)2^{-4k+4}10^k\le 200\ell^{-4}.\]
\end{proof}
To invoke Lemma \ref{runtime}, we also need to satisfy its condition (4.):
\begin{lemma}\label{longestpathpolylog}
With high probability in $m$, the point at which $\tau$ is reached comes before any point at which the corruption graph has a path of length at least $\log^{1.5}(m)$.
\end{lemma}
\begin{proof}
We may assume that Lemma \ref{numcorruptbins} holds (since it holds with high probability in $m$). Then, as proven in the proof of Lemma \ref{CorruptComponentSize}, we see that the expected number of paths of length $\log^{1.5}(m)$ is at most \[2m\ell^{-4(\log^{1.5}(m))}10^{\log^{1.5}(m)}\le e^{-\log^{1.2}(m)}.\](using that $\ell=O(\log m)$, $\ell\ge 2$, and $m$ sufficiently large). Then Markov's inequality gives that the probability of any path of length $\log^{1.5}(m)$ is at most $e^{-\log^{1.2}(m)}$, which gives our result with high probability in $m$.
\end{proof}

Finally, we prove that, indeed, by the time $\tau$ occurs, we have (with high probability) reached a very high load factor. As noted earlier, this will allow to assume, up to that load factor, that $\tau$ has \emph{not yet} occurred, which will allow us to employ Lemmas \ref{nobicyclic} and \ref{CorruptComponentSize} to reason about the structure of the corruption graph.

\subsection{Run-Time Analysis}
Finally, we can put the pieces together to reason about the running time of the algorithm.

\begin{lemma}\label{firstalgruntime}
Define $T_i$ to be the number of evictions performed by the $i$-th insertion (this is zero if the insertion process fails prior to the insertion, and is the number of evictions to failure if it fails during the insertion). Let $A$ be the indicator for $\tau$ not being reached by the end of the $i$-th insertion. Then, supposing $i \le (1-\epsilon)\ell m$, we have $\E[T_i \cdot A] \le O(\epsilon^{-1})$. 

\end{lemma}

\begin{proof}
Note that we can assume that all lemmas that happen with probability at least $1-O(\ell^{-6}m^{-1})$ over the insertion process hold, as if not we can use the trivial run time bound of $3\ell m$ per object (as discussed in the algorithm definition and Lemma \ref{bicycliconlybarrier}) to add an $O(\ell^{-5})=o(1)$ term onto our expected runtime. If $i\le (1-\epsilon)\ell m$, then there are $\epsilon\ell m$ empty slots, so there must be at least $\epsilon m$ non-full bins.

Then, the lemma follows from Lemmas \ref{runtime}, \ref{numcorruptbins} \ref{nobicyclic}, \ref{CorruptComponentSize}, and \ref{longestpathpolylog}. Lemma \ref{nobicyclic} gives that the process (with probability $\ge 1-O(\ell^{-6}m^{-1})$) does not fail before $\tau$ is reached. Then, we have that Lemmas \ref{numcorruptbins}, \ref{CorruptComponentSize}, and \ref{longestpathpolylog}, along with the definition of $\tau$, give us that, at any point before $\tau$ is reached, the conditions of Lemma \ref{runtime} are satisfied with $C=\mb{E}_{\text{random bin $B$, random hashes}}[W(B)]\le 200\ell^{-4}$. The conclusion of Lemma \ref{runtime} then implies that our expected number of evictions is $\epsilon^{-1}+O(1+\ell^{-4}\epsilon^{-1})=O(\epsilon^{-1})$ when an $\epsilon$ fraction of bins are not full.

\end{proof}

Putting together Lemmas \ref{boundingepstilde}, \ref{nobicyclic}, \ref{epsstarfirst}, and \ref{firstalgruntime} give the following theorem:

\thmfirst*
\begin{proof}
The first part of this theorem statement comes directly from Lemmas \ref{boundingepstilde}, \ref{nobicyclic}, \ref{epsstarfirst}, and \ref{firstalgruntime}. For the expected run time, we can note that the case where our algorithm does not actually fail, but one of those lemmas fails, also has probability $O(\ell^{-6}m^{-1})$. (Really, it really has $O(m^{-10})$ probability by the note after Lemma \ref{nobicyclic}, but the $O(\ell^{-6}m^{-1})$ given by the lemma statement is enough here.) In that case, we can use the trivial run time bound of $3\ell m$ per object (as discussed in the algorithm definition and Lemma \ref{bicycliconlybarrier}) to add an $O(\ell^{-5})=o(1)$ term onto our expected runtime.

Now, to bound the total number of evictions, recall the definition of special evictions from Lemma \ref{runtime}. Every special eviction reveals a new second hash, so there can be at most $\ell m$ special evictions. Therefore, to show that the expected total number of evictions is $O(\ell m)$, we need to show that the expected total number of non-special evictions is $O(\ell m)$.

Note that for any insertion with $\epsilon\ell m$ slots remaining, the expected number of special evictions is $\epsilon^{-1}$ (noting again that $\epsilon\ell m$ free slots implies at least an $\epsilon$ fraction of non-full bins), as each special eviction has at least an $\epsilon$ probability of finishing the process. Therefore, Lemma \ref{runtime} implies that the expected number of non-special evictions is $O(1+C\epsilon^{-1})$. Lemma \ref{boundingepstilde} gives with high probability in $m$ that our algorithm (and in particular, any algorithm) will stop when there are at least $.7^\ell m$ uninserted balls (noting $.7<2/e$). Therefore, considering the insertion of the $i$-th from last of the $\ell m$ balls, we see that there must be at least $i/\ell$ free slots, so we get that\begin{align*}
\mb{E}[\text{non-special evictions}]&\le O\left(\sum_{i=.7^\ell m}^{\ell m}\left(1+C(i/(\ell m))^{-1}\right)\right)\\&\le O\left(\sum_{i=.7^\ell m}^{\ell m}\left(1+200\ell^{-3}i^{-1}m\right)\right)\say{Lemma \ref{CorruptComponentSize}}\\&\le O(\ell m)+O\left(\ell^{-3}m\sum_{i=.7^\ell m}^{\ell m}i^{-1}\right)\\&\le O(\ell m)+O\left(\ell^{-3}m\ln\left(\frac{\ell m}{.7^\ell m}\right)\right)\\&\le O(\ell m)+O(\ell^{-2}m)\\&\le O(\ell m)
\end{align*}as desired.
\end{proof}

Note that, for small $\epsilon$ ($\epsilon = e^{-\Theta(\ell)}$), one could improve the bound on expected evictions to $O(\epsilon^{-1}\ell^{-1})$. This is because we use in the proof of Lemma \ref{firstalgruntime} that each non-full bin has at most $\ell$ empty slots, when really, sufficiently close to $\epsilon^*$, each non-full bin has $2\pm O(1/\ell)$ empty slots in expectation (see the proof of Lemma \ref{expepstilde}). We will defer formalizing this observation until our second algorithm, where we will get the appropriate $\ell^{-1}$ factor in Theorem \ref{secondalgthm}.

\section{The Second Algorithm: One-Step Look-Ahead}\label{secondalgsection}

In this section, we will modify the insertion algorithm from Section \ref{firstalgsection} and prove that our new algorithm gives the following guarantees:
\begin{restatable}{theorem}{secondalgthm}\label{secondalgthm}
Assume $\ell\le 1.5\ln(m)$ is at least a sufficiently large positive constant. Choose any value $\delta\in[.99^\ell, 1]$. Our algorithm has at least a $1-O(\ell^{-6}m^{-1})$ probability of successfully inserting $(1-\epsilon)\ell m$ balls into the hash table, where $\epsilon=(1+\delta)(\mb{E}[\epsilon^*])$. For insertions at load factor $1 - \epsilon$ for $\epsilon\ge .9^\ell$, our algorithm performs $O(\epsilon^{-1})$ expected evictions per ball. If $\epsilon\le .9^\ell$, then our algorithm performs $O((\epsilon-\mb{E}[\epsilon^*])^{-1}\ell^{-1})\le O(\delta^{-1}(\mb{E}[\epsilon^*])^{-1}\ell^{-1})$ expected evictions per ball, where each eviction takes $O(\ell)$ time.
\end{restatable}

Lemmas in this section give Lemma \ref{boundingepsstar} as a corollary, which says that $\epsilon^*=(2/e)^\ell\poly(\ell)$ with high probability in $m$, so (noting $2/e<.9$) the condition $\epsilon\le .9^\ell$ does indeed apply during later insertions. Theorem \ref{secondalgthm} means you could choose to fill the hash table until there are $1.001(\mb{E}[\epsilon^*])\ell m$ empty slots remaining (already a much better load factor than the first algorithm), and the expected number of evictions would be $O(\epsilon^{-1}\ell^{-1})$ after $\epsilon\le .9^\ell$. Or you could choose, say, $\delta=\ell^{10}$ to get an expected number of evictions of $O(\ell^{10}(\mb{E}[\epsilon^*])^{-1})=\tilde O((\mb{E}[\epsilon^*])^{-1})$ with $(1+\ell^{-10})(\mb{E}[\epsilon^*])\ell m$ free slots remaining.

In the algorithm definition, we will see that the algorithm performs a one-step look-ahead: before each eviction, the algorithm may need to check whether up to $\ell$ other bins are full or not (the second hashes of some balls in the current bin). This counteracts the $\ell^{-1}$ factor in the expected number of evictions to give a run time of $O(\delta^{-1}(\mb{E}[\epsilon^*])^{-1})$. A careful implementation of our algorithm could likely avoid paying the full $\ell$ factor. 
\subsection{Algorithm Definition}
 In Section \ref{firstalgsection}, we ran the same algorithm until a given event $\tau$ occurred. Now, we have two similar events, $\tau_1$ and $\tau_2$, that change which algorithm we will perform.
\begin{definition}\label{firsttosecondalgdef}
We start by performing the first algorithm from the previous section. Let $\tau_1$ be the event that $1.4\ell m$ hashes have been revealed under that algorithm. Once $\tau_1$ has occurred, we switch to performing the below modified algorithm. Let $\tau_2$ be the event that $.2\ell m$ hashes have been revealed since $\tau_1$ (so $1.6\ell m$ total hashes revealed).
\end{definition}
For our lemmas, we will consider running the modified algorithm until $\tau_2$ occurs. For the final run-time analysis in the theorem statement of Theorem \ref{secondalgthm}, we consider instead stopping the algorithm when the fraction of free slots $\epsilon$ satisfies $\epsilon=(1+\delta)(\mb{E}[\epsilon^*])$ (which we will show is likely to happen before $\tau_2$).

After $\tau_1$ occurs, we choose more judiciously which first-hash balls we evict from a bin, if some exist. The modified algorithm is below. The difference between this and the previous algorithm comes from point (4.)~of the previous algorithm being split into points (4.)~and (5.)~here.
\begin{enumerate} 
    \item If $B$ has an empty slot, place $x$ into the first empty slot of $B$.
    \begin{itemize}
        \item If $x$ is on its second hash, this ``first empty slot'' is the $\ell$th slot (so the insertion of $x$ makes $B$ full), and there is another ball in $B$ that is on its first hash, then exchange the positions within $B$ of $x$ and the first ball in $B$ that is on its first hash.
    \end{itemize}
    \item If the first $\ell-1$ slots of $B$ all contain objects on their second hash, place $x$ into the last slot in $B$.
    \item If $x$ was just evicted from its slot under point (2.)~of this list and $x$ has previously been placed into $B$, place $x$ into the last slot in $B$
    \item If there are any balls in the first $\ell-1$ slots of $B$ that are on their first hash and whose second hash goes to a bin with an empty slot, then evict the first one of those balls and place $x$ into its place.
    \item Otherwise, there must be balls in the first $\ell-1$ slots of $B$ that are on their first hash. Evict the first one of those balls and place $x$ into its place.
\end{enumerate}
We also use the same failure condition as for the first algorithm: if, in a sequence of evictions using points (2.)~and (3.), the first bin $B$ in that sequence is visited three times, we declare failure. 

Call a ball $x$ \underline{pure first} if it is currently stored in $h_1(x)$ and has never been placed into $h_2(x)$. Notice that, during an eviction, to distinguish between Cases (4.)~or (5.), we do not need to fully reveal the second hash $h_2(x)$ of the pure first balls $x$ in the bin -- it suffices to reveal whether or not bin $h_2(x)$ is empty. We refer to this as \underline{checking} a ball's second hash, as opposed to \underline{revealing} that hash, which is what occurs when a ball $x$ is actually \emph{evicted} to its second bin $h_2(x)$ for the first time.

In Section \ref{firstalgsection}, it was necessary for the load factor analysis (Lemma \ref{epsstarfirst}) that every time we revealed a hash to a non-full bin, we did place that ball in that non-full bin. Here, it is important for the load factor analysis for us to note that we can perform our checks such that every time the check finds a second hash to a non-full bin, we do place that ball in that non-full bin. In other words, when evicting from a bin $B$ by points (4.)~and (5.), we check the pure first balls in $B$ in some order, and as soon as a second hash to an available bin is found, the rest of the slots in $B$ are not checked.

Analogous proofs to Section \ref{firstalgsection}, but with point (4.)~replaced with ``points (4.)~and (5.)'', quickly give the following lemma, analogues of Lemmas \ref{corruptreturn}, \ref{uniformpurefirst}, and \ref{bicycliconlybarrier}:

\begin{lemma}
If a ball $x$ is ever returned to its first bin, $h_1(x)$, after having been evicted from it in the past, then $h_2(x)$ was a corrupt bin at the time of $x$'s initial eviction from $h_1(x)$. Additionally, $x$ is placed into the last slot of $h_1(x)$ when returned to $h_1(x)$.
\end{lemma}
\begin{lemma}
Any ball $x$ evicted from $h_1(x)$ by point (4.)~or (5.)~has never previously been placed in $h_2(x)$.
\end{lemma}
\begin{lemma}
Every insertion either succeeds or declares failure. Moreover, if an insertion declares failure, then the corruption graph (at the point in time where failure is declared) contains a bicyclic component. Moreover, no insertion (even one that fails) will ever perform a sequence of (2.)~and (3.)~evictions in which it visits the same slot more than 3 times.
\end{lemma}

\subsection{Properties between $\tau_1$ and $\tau_2$}

Using a similar argument as in Lemma \ref{numcorruptbins}, we can prove that there are likely to be few non-full bins by the time $\tau_1$ occurs. We also extend this to give a lower bound on the number of non-full bins as well.
\begin{lemma}\label{fewemptyattau1}
Let $E$ be the event that fewer than $.95^\ell m$ bins have an empty slot remaining and let $F$ be the event that fewer than $.91^\ell m$ bins have an empty slot remaining. With high probability in $m$ over the insertion process, $\tau_1$ will not occur before $E$ occurs, but $F$ will not occur before $\tau_1$ occurs.
\end{lemma}
(In other words, assuming the algorithm does not fail before $\tau_1$, the fraction of non-full bins at $\tau_1$ will be between $.91^\ell$ and $.95^\ell$.)
\begin{proof}
Consider the first $1.4\ell m$ revealed hashes $R_1, R_2, \ldots, R_{1.4\ell m} \in [m]$. As noted in Lemma \ref{epsstarfirst}, because, whenever we reveal a hash to a non-full bin, we place that ball in that bin, we have that the number of non-full bins at $\tau_1$ is exactly the number $X$ of bins $B$ satisfying $\sum_i \mathbf{1}_{R_i = b} < \ell$. Thus, it suffices to argue that, with high probability in $m$, $X$ is at most $.95^\ell m$ and at least $.91^\ell m$.

The distribution of $\sum_i \mathbf{1}_{R_i = b}$ is $\mathrm{Bin}(1.4\ell m,\frac 1m)$, which, by Lemma \ref{poisson}, satisfies
\[\mb{P}\left(\mathrm{Bin}\left(1.4\ell m,\frac 1m\right)<\ell\right)\le e^{-\ell(1.4-1-\ln(1.4))}<.94^\ell.\]
Similarly, analogous calculations to Lemma \ref{poisson} (standard Poisson-approximation bounds) give that
\[\mb{P}\left(\mathrm{Bin}\left(1.4\ell m,\frac 1m\right)<\ell\right)>.92^\ell.\]
The expected value of $X$ is therefore $.92^\ell m\le\E[X] \le .94^\ell m$. Notice that $X$ is a function of the independent random variables $R_1, R_2, \ldots, R_{1.4\ell m}$, and that changing any one $R_i$ changes $X$ by at most 1. Thus we can apply McDiarmid's inequality to deduce that \[
\mb{P}(|X-\mb{E}[X]|\ge.91^\ell m)\le 2e^{-\frac{(.91^\ell m)^2}{1.4\ell m}}\le 2e^{-.8^\ell m}
\]
for sufficiently large $\ell$. Noting that we are assuming that $\ell\le 1.5\ln(m)$, we have that $.8^\ell\ge m^{1.5\ln(.8)}\ge m^{-.4}$, giving that $e^{-.8^\ell m}\le e^{-m^{.6}}$, so we have $|X-\mb{E}[X]|\le.91^\ell m$ with high probability in $m$. Then \[.92^\ell m\le\E[X] \le .94^\ell m\text{~~~~and~~~~}|X-\mb{E}[X]|\le.91^\ell m\]together imply the lemma for sufficiently large $\ell$.
\end{proof}

One difficulty in analyzing the new algorithm will be the following subtlety: After $\tau_1$, once we swap to the new algorithm, insertions begin to make a concerted effort to evict elements whose second hash goes to a bin with a free slot -- so such bins are \emph{unusually} likely to accumulate new second-hash elements. This means that, ironically, these bins may also be more likely to become (subsequently) corrupt. 

To handle this issue in the analysis, we modify the definition of the corruption graph so that, when $\tau_1$ occurs, we \emph{preemptively} declare any bin that has at least one free slot to be corrupt. With his modification in place, we can recover the following variation of Lemma \ref{lem:Gprime} (although, as we will see, the proof is somewhat more intricate than before):



\begin{lemma}
Consider a threshold $T\le m/\ell$, and let $G$ be the corruption graph at the final point in time during which it has $\le T$ corrupt vertices, or when $\tau_2$ occurs (whichever comes first). Then, there exists a graph $G'$ that is constructed via the following process, and where the edges (resp.~corrupt vertices) in $G$ are a subset of the edges (resp.~corrupt vertices) in $G'$. 

The graph $G'$ is constructed by selecting a uniformly random set of (with-high-probability in $m$) at most $T + me^{-\ell/30}$ (distinct) vertices to be corrupt, and then adding (unordered) edges to each corrupt vertex $B$ as follows:  
\begin{enumerate}
    \item For each bin $B'$, $\mathrm{Bin}(\ell,\frac{1.1}{m})$ edges are added from $B$ to $B'$. This process occurs independently for different $B'$.
    \item $\mathrm{Bin}(\ell m,\frac{1.1}{m})$ edges of the form $(B, B')$ are added, where each $B'$ is independent and uniformly random.
\end{enumerate}
The two steps above are independent both of each other, and across all corrupt vertices.
\label{lem:Gprime2}
\end{lemma}
\begin{proof}
Define for each bin $b$:
\begin{itemize}
    \item $x_b$ to be the total number of balls $u$ inserted so far with $h_1(u) = b$, and where, \emph{at the time immediately prior to insertion}, $h_2(u)$ was not yet corrupt.
    \item $y_b$ to be the total number of balls $u$ inserted so far with $h_2(u) = b$, and where ball $u$ has, at some point, used its second hash (i.e., been placed in bin $b$ using $h_2$). 
\end{itemize}

Whereas, in the proof of Lemma \ref{lem:Gprime}, the system state consisted of $(G, G', \{x_b\})$, the system state will now be $(G, G', \{x_b\}, \{y_b\})$ (where, as before, $b$ ranges over only the remaining non-corrupt bins). 

We note that, in the proof of Lemma \ref{lem:Gprime} (specifically, in the handling of Type 1 balls) we made use of the following fact: when a bin $B$ becomes corrupt, each ball $x$ in each bin $B'$ that has not yet used its second hash (and that does not already correspond to an edge in $G$) independently has probability $1/K$ of satisfying $h_2(x) = B$, whre $K$ is the number of remaining non-corrupt bins prior to $B$ becoming corrupt. We claim that this fact is still true (but a bit more subtle than it was before). We must be careful because, even if a ball $x$ with $h_2(x)$ has never-yet \emph{used} its second hash, it is possible that it has \emph{checked} it (after $\tau_1$ occurs). Notice, however, that after $\tau_1$ occurs, if a ball $x$ has not yet used its second choice and doesn't correspond to an edge in $G$, then $h_2(x)$ cannot have any free slots (since all bins with free slots are corrupt). Therefore, whenever the algorithm \emph{checks} the second hash of such a ball, the check is guaranteed to fail, revealing no information about which of the (not-yet-corrupt bins) $h_2(x)$ is. As such, it remains the case that $h_2(x)$ is uniformly random out of the not-yet-corrupt bins. 

With this point in mind, we can begin the construction of $G'$ exactly as in Lemma \ref{lem:Gprime}. As vertices become corrupt in $G$, we declare the same vertices to be corrupt in $G'$, and add edges as in the proof of Lemma \ref{lem:Gprime}. The resulting graph $G'$ (which we are not done adding vertices and edges to) has the property that, when a corrupt vertex is introduced, it gets edges according to the distribution described in the lemma statement (independent of the system state). 

Finally, the final step in the proof is to add additional vertices (and edges for those vertices) to $G'$ in order ensure that the fianl set of corrupt vertices is uniformly random and independent of the edges (while also ensuring that the total number of corrupt vertices is at most $T + O(me^{-\ell/30})$). This is the part of the proof that requires more care than in the proof of Lemma \ref{lem:Gprime}.

To describe which additional vertices (and edges) to add to $G'$, and to aid in our analysis, it will be helpful to imagine that the randomness in the construction comes from three sources: 
\begin{itemize}
    \item $R_1$: a tape used to generate the randomness for adding edges to $G'$ once a vertex $b$ becomes corrupt in $G$. This includes determining for each element $x$ that has not yet revealed its second hash (but that is also not yet an edge) whether $h_2(x) = b$; determining $h_1(x)$ for each element $x$ that is being inserted in the future and that has second hash $h_2(x) = b$; and determining randomness for edges that don't actually correspond to elements (i.e., for the edges $G'$ has that $G$ does not have). 
    \item $R_2$: a tape of random numbers in $[m]$ used as follows. Whenever an element $x$ is inserted, if it does not already correspond to an edge in $G'$ (i.e., $h_2(x)$ is not already corrupt), it selects its first hash $h_1(x)$ from $R_2$. 
    \item $R_3$: a tape of random numbers in $[m]$ used as follows: Whenever an element $x$ that is not yet an edge in $G'$ reveals its second hash (which is therefore uniformly random out of the not-yet-corrupt nodes), the second hash is determined by reading numbers off of $R_3$ until we get a not-yet-corrupt node. 
\end{itemize}

It will be important to calculate the number of times $q_2$ and $q_3$ that we read from each of $R_2$ and $R_3$ before $\tau_1$ occurs (if $\tau_1$ does not occur, define $q_1, q_2 =\infty$, as that makes the following analysis go through smoothly). Lemma \ref{fewemptyattau1} tells us that, by the time $\tau_1$ occurs, the load factor of the hash table is at least $1 - .95^\ell$ (with high probability in $m$); the fact that $T \le m/\ell$ tells us that, at any time during the construction, $q_2$ is (with high probability) a sum of at least $(1 - .95^\ell)\ell m$ indicator RVs that are each independently $1$ with probability at least $1 - 1/\ell$. This implies by a Chernoff bound that $q_2 \ge r_2$ for $r_2 = (1- 2/\ell) \ell m$, with high probability in $m$ using $\ell=O(\log m)$ (note that, if $\tau_1$ does not occur, then $q_2 := \infty$ also satisfies $q_2 \ge r_2$).  On the other hand, by the time $\tau_1$ occurs, $1.4 \ell m$ hashes have been revealed (by definition), at least $.4 \ell m$ of which are second hashes; since there are at most $m/\ell$ corrupt nodes, we have (with high probability in $m$) that each of these hashes independently has probability $(1 - 1/\ell)$ of being a not-yet-corrupt node (and using $R_3$); it follows by a Chernoff bound that, with high probability in $m$, that $q_3 \ge r_3$ for some $r_3 = .4 \ell m (1 - 2/\ell)$ (as before, if $\tau_1$ doesn't occur, then $q_3 := \infty$ trivially satisfies this inequality). 

Let $\overline{R}_2$ be the first $r_2$ entries of $R_2$, and let $\overline{R}_3$ be the first $r_3$ entries of $R_3$. With high probability, all of $\overline{R}_2$ and $\overline{R}_3$ are read before $\tau_1$ (or $\tau_1$ never occurs). Every time a number $j$ is read from $\overline{R}_2$ (resp. $\overline{R}_3$), then either bin $j$ is already corrupt, or we increment $x_j$ (resp.~$y_j$) by 1. It follows that, with high probability in $m$, if $\tau_1$ occurs, then for each bin $j$ that is not already corrupt by the time $\tau_1$ occurs, we have (when $\tau_1$ occurs) that 
$$x_j + y_j \ge \sum_{i = 1}^{r_2} \mathbf{1}_{R_2[i] = j} + \sum_{i = 1}^{r_3} \mathbf{1}_{R_3[i] = j}.$$
Define $Q$ to be the set of bins that become corrupt at $\tau_1$ (bins that are not corrupt prior to $\tau_1$ and that have at least one free slot at $\tau_1$). Such bins must satisfy $x_j + y_j < \ell$, meaning that every bin $j \in Q$ satisfies 
$$\sum_{i = 1}^{r_2} \mathbf{1}_{R_2[i] = j} + \sum_{i = 1}^{r_3} \mathbf{1}_{R_3[i] = j} < \ell.$$
Define $\overline{Q}$ to be the set of bins $j$ that satisfy the inequality above. With high probability in $m$, $Q \subseteq \overline{Q}$. 

Additionally, define $W$ to be the set of all bins $b$ that, prior to $\tau_2$ during the construction of $G$, experience $x_b \ge \ell- 1$. By construction, $W \cup Q$ contains all bins that become corrupt in $G$. Finally, define $\overline{W}$ to be the set containing the first $T$ numbers to appear $\ell - 1$ times in $R_2$ (i.e., if we read $R_2$ and record when a number has appeared $\ell - 1$ times, then the first $T$ numbers we record are in $\overline{W}$). Since $G$ has at most $T$ corrupt vertices, we can deduce that $W \subseteq \overline{W}$. 

Putting the pieces together, so far, we can conclude that, with high probability in $m$, the set of corrupt bins in $G$ is a subset of $\overline{W} \cup \overline{Q}$. To complete our construction of $G'$, for each bin $b$ in $\overline{W} \cup \overline{Q}$ that is not corrupt in $G$, declare it to be corrupt in $G'$ and add edges to $G'$ as specified in the lemma statement (to get the randomness for these edges, one can use further randomness from $R_1$). This completes the construction of $G'$.

To complete the proof of the lemma, we must argue two things: (1) that the random bits used to select the corrupt vertices in $G'$ is independent of the randomness used to determine the edges that are added to $G'$; (2) that, the total number of corrupt vertices in $G'$ is, with high probability in $m$, at most $T + O(me^{-\ell/30})$. 

By construction, the set of corrupt vertices in $G'$ is fully determined by the randomness in $R_2$ and $R_3$, while the edges that are added to $G'$ use randomness exclusively from $R_1$. Therefore, by construction, the set of vertices is independent of the set of edges. (Also, by symmetry between vertices, the set of corrupt vertices is a random subset of all bins.)

Finally, we bound the number of corrupt vertices in $G'$. By construction, $|\overline{W}| = T$, so it suffices to prove (with high probability in $m$) that
$$|\overline{Q}| \le O(m e^{-\ell/30}).$$
Recall that $\overline{Q}$ is the set of bins that appear fewer than $\ell$ times in $\overline{R}_2$ and $\overline{R}_3$, which together have size $r_2 + r_3 \ge (1- 2/\ell) 1.4 \ell m$. The expected number of such bins is 
$$m \Pr[\mathrm{Bin}((1-2/\ell) 1.4 \ell m, 1/m) < \ell],$$
which by Lemma \ref{poisson} is at most $m e^{-\ell((1-2/\ell).4)^2/3}\le me^{-\ell/20}$. Since the number of such bins is determined by the $O(\ell m)$ random numbers in $\overline{R}_2$ and $\overline{R}_3$, and since each random entry in $\overline{R}_2$ and $\overline{R}_3$ can change $|\overline{Q}|$ by at most 1, we can apply McDiarmid's inequality to deduce that 
$$\Pr[|\overline{Q}| > \E[|\overline{Q}|] + k \sqrt{\ell m}] \le e^{- \Omega(k^2)}.$$
Using $k = \log m$, we get that 
$$\Pr[|\overline{Q}| > \E[|\overline{Q}|] + \log m \sqrt{\ell m}] \le 1/2^{\omega(\log m)}.$$
Since $\E[|\overline{Q}|] \le m e^{-\ell/20}$, this gives that with high probability in $m$, 
$$|\overline{Q}| \le \E[|\overline{Q}|] + \log m \sqrt{\ell m} = m e^{-\ell/20} + \log m \sqrt{\ell m} \le m e^{-\ell/30},$$
where the final step uses that $\ell = O(\log m)$. This completes the proof of the lemma. 
\end{proof}

Lemma \ref{lem:Gprime2} will play the exact same role in our analysis as Lemma \ref{lem:Gprime} did in the previous algorithm. In particular, the analogue of Lemma \ref{probedgeexists} goes through with the exact same proof:

\begin{lemma}
Consider the state of the dominating corruption graph $G'$, described in Lemma \ref{lem:Gprime2}, at any given moment. Any edge (including any loop) has probability at most $\frac{5\ell}{m}$ of existing in the graph, and this remains true when conditioning on which vertices are corrupt. Similarly, the probability of a multi-edge existing with multiplicity at least $c$ is at most $\left(\frac{5\ell}{m}\right)^c$, again regardless of which vertices are corrupt; and the probability of any $c'$ edges (possibly with repeats) all existing is at most $\left(\frac{5\ell}{m}\right)^{c'}$, again regardless of which vertices are corrupt.
\end{lemma}

Additionally, we will need the following Lemma \ref{adapted1}, which bounds the total number of bins that become corrupt at any point before $\tau_2$ (this lemma plays the role of Lemma \ref{numcorruptbins} in the previous section).

\begin{lemma}\label{adapted1}
Let $E$ be the event that $.96^\ell m$ bins are corrupt. With high probability in $m$, $E$ will not occur before $\tau_2$.
\end{lemma}
\begin{proof}
Throughout the lemma, we will focus on the point in time immediately before $\tau_2$ occurs (or at the end of the insertion process if $\tau_2$ never occurs). We will argue that, with high probability in $m$, the number of corrupt bins is at most $.96^\ell m$.

We must first bound the number of bins that become corrupt when $\tau_1$ occurs, i.e., bins that contain a free slot at that time. By Lemma \ref{fewemptyattau1}, the number of such bins is, with high probability in $m$, at most $.95^\ell m$.

Next we bound the number of bins that become corrupt by virtue of collecting $\ell -1 $ elements that use the bin as a second hash. For our analysis, it will be useful to assume that there are two separate random tapes which each contain uniformly random bins in $[m]$: the first that first hashes read off of and the second that second hashes read off of. When a second hash is checked, we simply choose from a third source of randomness whether the check will show an non-full bin or a full bin. Then, when we reveal a hash that has already been checked to be a full bin, we read its bin off the second tape, but employ rejection sampling until we find one that is full. At any time after $\tau_1$, the number of bins with free slots is at most $(.95^\ell)m$ by Lemma \ref{fewemptyattau1}, so knowing that a hash has been checked to be full multiplies the probability of $h_2(u)$ being a given bin by at most a $1 + O(.95^\ell) \le 1.01$ factor.


Then doing a similar calculation as Lemma \ref{numcorruptbins}, we note that, prior to $\tau_2$, we have at most $1.6\ell m$ hashes revealed, so the number of second hashes revealed is at most $.6\ell m$. Then the number of second hashes to a given bin (that was not empty at $\tau_1$) is stochastically dominated by $\mathrm{Bin}(.6\ell m,1.01/m)$, which is stochastically dominated by $\mathrm{Bin}(.65\ell m,1/m)$ for sufficiently large $m$. Then we can apply Lemma \ref{poisson} to show
\[\mb{P}(\mathrm{Bin}(.65\ell m,1/m)\ge\ell-1)\le e^{-\ell(.65-1-\ln(.65))}<(.94)^{\ell},\]
and if we let $X$ be the number of corrupt bins (that did not have a free slot at $\tau_1$), we see that $\mb{E}[X]\le(.94)^\ell m$. As in Lemma \ref{numcorruptbins}, we note that changing any second hash location (from the second tape) can change $X$ by at most one, so we apply McDiarmid's Inequality to say that \[\mb{P}(X\ge .95^\ell m)\le\mb{P}(X-\mb{E}[X]\ge .94^\ell m)\le e^{-\frac{((.94)^\ell m)^2}{.6\ell m}}\le e^{-.8^\ell m}.\]
As noted in the proof of Lemma \ref{fewemptyattau1}, $1-e^{-.8^\ell m}$ is with high probability in $m$. Therefore, we have with high probability in $m$ that the total number of corrupt bins is at most $.95^\ell m + .95^\ell m \le .96^\ell m$, as desired.
\end{proof}

Now, we see from Lemma \ref{adapted1} that we can hope to apply Lemma \ref{lem:Gprime2} with threshold $T=.96^\ell m$, giving $T + me^{-\ell/30}\le \ell^{-10}m$ (for sufficiently large $\ell$) corrupt vertices in the creation of $G'$. This serves as a direct analogue of Lemma \ref{lem:Gprime}, so now the analysis of the corruption graph in the previous section (Lemmas \ref{nobicyclic}, \ref{CorruptComponentSize}, and \ref{longestpathpolylog}) continues to hold with exactly the same proofs:

\begin{lemma}\label{adapted2}
With probability at least $1-O(\ell^{-6}m^{-1})$, our algorithm will not fail before $\tau_2$ occurs.
\end{lemma}
\begin{lemma}\label{corruptcomp2}
Consider the corruption graph $G$ when $\tau_2$ occurs (or if the insertion process fails before $\tau_2$, let $G$ be the corruption graph before the failed insertion). Choose a bin $B$ uniformly at random and let $W(B)$ be the longest path in $G$ that contains $B$. Then \[
\mb{E}[W(B)]\le 200\ell^{-4}.
\]
\end{lemma}
\begin{lemma}\label{polylog2}
With high probability in $m$, the point at which $\tau_2$ is reached comes before any point at which the corruption graph has a path of length at least $\log^{1.5}(m)$.
\end{lemma}

Next we prove that, by the time $\tau_2$ occurs, we are very likely to have visited \emph{almost all} of the bins, and to have at least \emph{checked} the second hashes of almost all elements. This lemma will be critical, as it will allow us to subsequently argue that, by $\tau_2$, our algorithm will have reached a very high load factor (much higher than we could get using the argument in the proof of Lemma \ref{fewemptyattau1}).
\begin{lemma}\label{fewunchecked}
Let $E$ be the event that at most $.99^\ell m$ out of the total $\ell m$ balls have not had their second hash checked or revealed. With probability at least $1-O(\ell^{-6}m^{-1})$, $E$ will happen before $\tau_2$.
\end{lemma}
\begin{proof}
We can assume that Lemma \ref{adapted2} holds (the algorithm does not fail before $\tau_2$), which happens with probability at least $1-O(\ell^{-6}m^{-1})$.

For $\tau_2$ to occur, and for some ball $x$ to not have its second hash checked or revealed, one of the following five conditions must hold:
\begin{itemize}
    \item $x$ is not in the table at $\tau_1$
    \item $h_1(x)$ has an empty slot at $\tau_1$
    \item $h_1(x)$ is corrupt by $\tau_2$ (technically, this condition contains the previous one)
    \item There is a ball $y$ (possibly $y = x$) such that $h_1(x)=h_1(y)$, $h_2(y)$ had an empty slot at $\tau_1$, and $y$ was not in $h_2(y)$ at $\tau_1$
    \item $h_1(x)$ is never visited between $\tau_1$ and $\tau_2$, except by balls that had previously been placed in $h_1(x)$.
\end{itemize}
First let us explain why, if none of these five occur, the second hash of $x$ must get revealed or checked. If $x$ is ever on its second hash then its second hash has already been revealed or checked, and we are done. Otherwise, since $x$ is in the table by $\tau_1$ (by bullet 1 not occurring), it must be in $h_1(x)$. If bullet five does not occur, then $h_1(x)$ is visited between $\tau_1$ and $\tau_2$ by a ball $z$ that has never before been placed into $h_1(x)$. That ball\begin{itemize}
    \item does not evict by point (1.)~(since the second bullet is false, so $h_1(x)$ has no empty slots)
    \item does not evict by point (2.)~(since the third bullet is false, so $h_1(x)$ is not corrupt)
    \item does not evict by point (3.)~(since $z$ has never before been placed into $h_1(x)$)
    \item does not evict by point (4.)~(since the fourth bullet is false, so no pure first ball in $h_1(x)$ has its second hash going to an empty slot)
\end{itemize}Therefore, $h_1(x)$ gets an eviction by point (5.), which requires (due to point (4.)) that the second hash of every ball in $h_1(x)$ must be checked.

So, now we know that one of the five bullet points must occur for every ball that does not have its second hash checked or revealed. We can now go through each bullet point and upper bound the number of balls that have that bullet point occur.

By Lemma \ref{fewemptyattau1}, with high probability in $m$ there are at most $.96^\ell m$ non-full bins at $\tau_1$, and thus $\ell(.96^\ell m)$ balls that have not yet been inserted at $\tau_1$ (the first bullet).

The second bullet is a special case of the third bullet, and the number of bins due to the third bullet is bounded by $.96^\ell m$ with high probability in $m$ by Lemma \ref{adapted1}.

For the fourth bullet, we first assume that Lemma \ref{fewemptyattau1} holds. We claim that, conditioned on that, with probability at least $1-\left(.97\right)^{.98^\ell m}$ the number of balls with any hash to a non-full bin is at most $.98^\ell m$. In particular, the probability that any set of $.96^\ell$ bins has more than $.98^\ell m$ hashes to them is upper bounded by \begin{align*}
\binom{m}{.96^\ell m}&\binom{2\ell m}{.98^\ell m}\left(\frac{.999^\ell m}{m}\right)^{.98^\ell m}\\&\le\left(\frac{em}{.96^\ell m}\right)^{.96^\ell m}\left(\frac{2e\ell m}{.98^\ell m}\right)^{.98^\ell m}\left(\frac{.96^\ell m}{m}\right)^{.98^\ell m}
\\&\le\left(\left(\frac{e}{.96^\ell }\right)^{(.96/.98)^\ell}\left(\frac{2e\ell}{.98^\ell}\right)\left(.96\right)\right)^{.98^\ell m}
\\&\le\left((1.0001)(1.0001)(.96)\right)^{.98^\ell m}\\&\tag{for sufficiently large $\ell$, using $\lim_{x\rightarrow\infty}(cx)^{x^{-d}}=1$ for any constants $c,d>0$}
\\&\le\left(.97\right)^{.98^\ell m}.
\end{align*}
Then, we see that $1 - \left(.97\right)^{.98^\ell m}$ is with high probability in $m$ as $.98^\ell m\ge .98^{1.5\ln(m)}m\ge m^{.9}$.

Finally, for the fifth bullet, we need to count the number of bins that are not visited by new hashes between $\tau_1$ and $\tau_2$. There are $.2\ell m$ new hashes between $\tau_1$ and $\tau_2$. The only possible conditionings of these hashes before they are revealed is whether or not they will go to a non-full bin. The previous calculation showed that with high probability in $m$, at most $.98^\ell m$ hashes go to the set of open bins. So regardless of any conditioning, at least $.2\ell m-.9999^\ell m\ge .1\ell m$ hashes must go to a uniformly random bin out of ones without an empty slot.

Once we are in that with high probability case, the probability that at least $.98^{\ell}m$ bins will not get any of those $.1\ell m$ hashes is \begin{align*}
\binom{m}{.98^\ell m}&\left(1-\frac{.98^\ell m}{m}\right)^{.1\ell m}
\\&\le\left(e(.98^{-\ell})\right)^{.98^\ell m}e^{-.99^\ell(.1\ell m)}
\\&\le\left(.98^{-\ell}\right)^{.99^\ell m}e^{-.99^\ell(.1\ell m)}
\\&\le\left((.98^{-1})e^{-.1}\right)^{\ell(.99^\ell m)}
\\&\le\left(.93\right)^{\ell(.99^\ell m)},
\end{align*}
which again is enough for a high-probability bound.

To sum it up, we have that, with high probability in $m$ (conditioned on Lemma \ref{nobicyclic} holding), the total number of balls that satisfy one of the five conditions (which upper bounds the total number of balls that do not have their second hash checked by $\tau_2$) is at most $\ell(.96^\ell)m+\ell(.96^\ell)m+.97^\ell m+.98^\ell m+\ell(.98^\ell m)$, which for sufficiently large $\ell$ is at most $.99^\ell m$.
\end{proof}

Intuitively, if we have checked (or revealed) the second hash of almost every element, then we must be at a very high load factor. The following lemma captures this. Note that the lemma refers to $\tilde\epsilon$, which is defined in Section \ref{epsstarsection} and satisfies $\tilde\epsilon \le \epsilon^*$. 
\begin{lemma}\label{epsstarsecond}
Let $\epsilon_{\tau_2}$ be the random variable such that $\epsilon_{\tau_2}\ell m$ is the number of free slots in the table at the time when the event $\tau_2$ occurs (if the algorithm fails before $\tau_2$ occurs, set $\epsilon_{\tau_2}=\tilde{\epsilon}$). Then with high probability in $m$, we have $\epsilon_{\tau_2}\le\tilde\epsilon/(1+.99^\ell)$.
\end{lemma}
This proof will follow a similar outline to Lemma \ref{epsstarfirst}, except that we can now use Lemmas \ref{fewunchecked} and \ref{fewemptyattau1} to get a better bound.

Note that this lemma says that for arbitrarily large $\ell$, we can get within a factor of 1.0001 to the optimal $\epsilon$. In this sense, this is much stronger than Lemma \ref{epsstarfirst}.
\begin{proof}
It suffices to bound the probability that $\tau_2$ occurs, and that when $\tau_2$ occurs, we have more than $\tilde\epsilon/(1+.99^\ell)$ free slots. 

Every time that we check or reveal a hash to a bin that has empty slots remaining, we always put that ball in the bin (this is true for both first and second hashes); and we do not evict a ball from a bin unless that bin is full. So, once $\tau_2$ has occurred, the only balls $x$ that have at least one hash referencing a bin with empty slots (and that are not already in that bin) are balls $x$ with at least one hash (namely, their second hash) that has not yet been checked/revealed. With probability at least $1-O(\ell^{-6}m^{-1})$, Lemma \ref{fewunchecked} bounds the total number of such balls by $.99^\ell m$. Call these balls the \underline{unchecked balls}.

The difference between $(1-\epsilon_{\tau_2})\ell m$ and $(1-\tilde\epsilon)\ell m$ is upper bounded by the number of unchecked balls that, after $\tau_2$ occurs, have at least one unchecked/unrevealed hash pointing at a bin with a free slot. Each unchecked ball independently has probability at most $2 \epsilon_{\tau_2}$ of having a unchecked/unrevealed hash pointing at a bin with a free slot. Therefore (for a given $\epsilon_{\tau_2}$, and assuming there are at most $.99^\ell m$ unchecked balls), we have that 
$$\epsilon_{\tau_2} \ell m - \tilde\epsilon \ell m$$
is dominated by $\mathrm{Bin}(.99^\ell m, 2 \epsilon_{\tau_2})$. By a Chernoff bound, we have with high probability in $m$ that such a binomial random variable is at most $.99^\ell m \cdot 2 \epsilon_{\tau_2} + \tilde{O}(\sqrt{m})$. It follows that, with high probability in $m$,
$$\epsilon_{\tau_2} \ell m - \tilde\epsilon \ell m \le .99^\ell m \cdot 2 \epsilon_{\tau_2} + \tilde{O}(\sqrt{m}).$$
This rearranges to 
$$\epsilon_{\tau_2} - \tilde\epsilon \le .99^\ell \cdot 2 \epsilon_{\tau_2} \ell^{-1} + \tilde{O}(1/(\ell\sqrt{m})),$$
which further rearranges to 
\begin{equation}\epsilon_{\tau_2} \cdot (1 - .99^\ell \cdot 2 / \ell) \le \tilde\epsilon + \tilde{O}(1/(\ell\sqrt{m})) \le \tilde\epsilon + m^{-.49},
    \label{eq:fooa}
\end{equation}
assuming $m$ is sufficiently large.

Finally, since $\ell \le 1.5\ln m$, we have that $(.99)^\ell\ge m^{1.5\ln(.99)}=\Omega(m^{-.01})$. Therefore, with high probability in $m$, we have
\begin{equation}
\tilde\epsilon + m^{.49} \le \tilde\epsilon \cdot (1 + O((.99)^\ell / m^{.1})).
\label{eq:foob}
\end{equation}
Supposing $\ell$ and $m$ are sufficiently large, \eqref{eq:fooa} and \eqref{eq:foob} combine to give that \[\epsilon_{\tau_2} \cdot (1 - .99^\ell \cdot 2 / \ell)\le\tilde\epsilon \cdot (1 + O((.99)^\ell / m^{.1})),\]implying that \[\epsilon_{\tau_2}\le\tilde\epsilon/(1+.99^\ell)\]as desired.

\end{proof}

\subsection{Run-time analysis}

Finally, we can now analyze running time in a similar way to the earlier algorithm, but where we now benefit from the fact that, by the time $\tau_2$ occurs, our algorithm is guaranteed to have reached a very high load factor. During the running of this algorithm, define $\epsilon$ such that $\epsilon\ell m$ slots are currently open.

The next two lemmas consider, at any time after $\tau_1$ has occurred but before $\tau_2$, the elements (including those not yet inserted) that are capable of filling a currently-free slot. The lemmas establish that the majority of these elements (1) are already in the hash table (Lemma \ref{fewoutsidehashed}) and (2) are the only such element in the bin where they currently reside (Lemma \ref{fewdoubles}). This will allow us to argue in the insertion-analysis that each time we check a hash for the first time, we have probability $\Omega(\epsilon)$ of finding a bin with a free slot. 

\begin{lemma}\label{fewoutsidehashed}
With high probability in $m$ over the insertion process, there is no point when $\tau_1$ has occurred, $\tau_2$ has not, and the number of hashes from balls outside the hash table to a bin that is currently non-full is at least $.99^\ell\epsilon\ell m/5$.
\end{lemma}
\begin{proof}
Consider some insertion after $\tau_1$ has occurred but before $\tau_2$ has occurred. By Lemma \ref{fewemptyattau1}, we have with high probability in $m$ that at most $.95^\ell m$ balls are not inside the hash table. Then we can use the fact that each unrevealed hash has probability at most $\epsilon$ of landing on a bin with a free slot (independently of anything that has happened in our algorithm so far) to say that the probability that more than $.99^\ell\epsilon\ell m/4$ hashes outside our table point to a currently non-full bin is at most
\[
\binom{.95^\ell m}{.99^\ell\epsilon\ell m/4}(2\epsilon)^{.99^\ell\epsilon\ell m/4},
\]as the first factor gives the number of ways to choose $.99^\ell\epsilon\ell m/4$ of the balls outside the table; and the second factor gives the probability, for each of those balls, that either of their two hashes goes to a currently non-full bin. Then we have that
\begin{multline*}
\binom{.95^\ell m}{.99^\ell\epsilon\ell m/4}(2\epsilon)^{.99^\ell\epsilon\ell m/4}\le\left(\frac{.95^\ell m(8e\epsilon)}{.99^\ell\epsilon\ell m}\right)^{.99^\ell\epsilon\ell m/4}\\\le\left(\frac{.95^\ell(8e)}{.99^\ell\ell}\right)^{.99^\ell\epsilon\ell m/4}\le(.97^\ell)^{.99^\ell\epsilon\ell m/4}.
\end{multline*}
Recall by assumption that $\E[\tilde\epsilon]=\Omega(m^{-.47})$, which implies by Lemma \ref{boundingepstilde} that $\epsilon \le \tilde\epsilon=\Omega(m^{-.47})$ with high probability in $m$.
Additionally note (as in Lemma \ref{epsstarsecond}) that $\ell^2(.99^\ell)\ge m^{1.5\ln(.99)} O(\log^2 m)\ge\Omega(m^{-.01})$. Thus $(.97^\ell)^{.99^\ell\epsilon\ell m/4}\le(.97)^{\Omega(m^{.52})}$, so we have with high probability in $m$ that at most $.99^\ell\epsilon\ell m/4$ hashes outside our table point to a currently non-full bin. 

The above analysis considered a specific insertion between $\tau_1$ and $\tau_2$ occurring. Union bounding over the up to $\ell m$ insertions that occur completes the proof.
\end{proof}
\begin{lemma}\label{fewdoubles}
With high probability in $m$ over the insertion process, there is no point when $\tau_1$ has occurred, $\tau_2$ has not, the number of empty slots is $\epsilon\ell m$ for some $\epsilon\ge\mb{E}[\tilde\epsilon]$, and the number of bins that currently contain two different pure first balls with unrevealed second hashes that go to a non-full bin is at least $.99^\ell\epsilon\ell m/5$.
\end{lemma}
\begin{proof}
Consider an insertion that takes place after $\tau_1$ occurs but before $\tau_2$ occurs. 
By Lemma \ref{fewemptyattau1}, we have with high probability in $m$ that $\epsilon\le.95^\ell$.

Let $X$ be the number of bins that currently contain two different pure first balls with unrevealed second hashes that go to a non-full bin. We claim that the probability that $X$ is at least $.99^\ell\epsilon\ell m/5$ is at most
\[
\binom{m}{.99^\ell\epsilon\ell m/5}(\ell^2)^{.99^\ell\epsilon\ell m/5}(\epsilon)^{2(.99^\ell\epsilon\ell m/5)}.
\]This is because the first factor chooses $.99^\ell\epsilon\ell m/5$ candidate bins; the second factor then chooses, for each of those bins, two of their pure first balls (of which each bin has at most $\ell$); and the third, given the $2(.99^\ell\epsilon\ell m/5)$ specified pure first balls, gives the probability that each of their second hashes goes to a non-full bin. Then we have that
\begin{align*}
\binom{m}{.99^\ell\epsilon\ell m/5}(\ell^2)^{.99^\ell\epsilon\ell m/5}(\epsilon)^{2(.99^\ell\epsilon\ell m/5)}&\le\left(\frac{5em\ell^2\epsilon^2}{.99^\ell\epsilon\ell m}\right)^{.99^\ell\epsilon\ell m/5}\le\left(\frac{5e\ell\epsilon}{.99^\ell}\right)^{.99^\ell\epsilon\ell m/5}\\&\le\left(.97^\ell\right)^{.99^\ell\epsilon\ell m/5}\le\left(.97\right)^{.99^\ell\epsilon\ell^2 m/5}.
\end{align*}
Which, as in Lemma \ref{fewoutsidehashed}, gives a with high probability in $m$ statement, and again we union bound over all $\le\ell m$ balls that we insert.
\end{proof}

Now, we have the tools to prove an equivalent of Lemma \ref{runtime} but for our modified algorithm. Note that the conclusion will look differently, with an extra factor of $(\epsilon-\E[\epsilon^*])^{-1}$:
\begin{lemma}\label{runtimesecondalg}
Assume that $\tau_1$ has occurred but $\tau_2$ has not. Let $G$ be the corruption graph immediately prior to the $k$-th insertion (or if one of the first $k$ insertions fail, let it be the corruption graph then), but with the edge corresponding to the $k$-th insertion removed (if present). Suppose we have already proven that:
\begin{enumerate}
    \item For a uniformly random bin $B$, the expected length of the longest path containing $B$ (and using each edge at most once) in $G$ is $\le C$ (here $G$ is a random variable). Here, path length is measured in number of edges. 
    \item Immediately prior to the insertion, we have with high probability in $m$ that at least an $\epsilon$ fraction of bins contain at least one free slot. 
    \item  Immediately prior to the insertion, there are at most $m/2$ corrupt bins, with high probability in $m$.
    \item In the corruption graph $\overline{G}$ after the insertion, the  longest path (anywhere in the graph) has length at most $\polylog m$ with high probability in $m$. 
    \end{enumerate}
Then, the expected number of evictions to complete the $i$-th insertion (where the number is $0$ if some previous insertion failed, and is the number of evictions to failure if the $i$-th insertion fails) is $10(\epsilon-\E[\epsilon^*])^{-1}\ell^{-1} + O(1 + C(\epsilon-\E[\epsilon^*])^{-1})\ell^{-1}$.
\end{lemma}
\begin{proof}
First, note that we can assume any lemmas that hold with high probability in $m$ do in fact hold here, as otherwise we can use the trivial run-time bound of $3\ell m$ per ball, since as discussed after the definitions of our two algorithms and in Lemma \ref{bicycliconlybarrier}, this adds an $o(m)$ summand to our expected number of evictions. Therefore, we assume that lemmas \ref{fewoutsidehashed} and \ref{fewdoubles} hold at every point when $\tau_1$ has occurred, and that the number of non-full bins when $\tau_1$ occurs is between $.91^\ell m$ and $.95^\ell m$ by Lemma \ref{fewemptyattau1}.

We maintain the definition of special evictions from the proof of Lemma \ref{runtime}, that is, an eviction is \underline{special} if it evicts a ball $x_i$ whose second hash has never yet been revealed and whose edge is not in the graph $G$. The first sentence in the proof of Lemma \ref{runtime} that does not continue to hold here is the line ``Note that $\Pr[b_j \text{ exists}]$ is at most the probability that the insertion reveals at least $j + 1$ fresh hashes (hashes never yet revealed before) without completing, which is at most $(1 - \epsilon)^{j + 1}$.'' Now, we might have some information about hashes before they are revealed - namely, whether they go to a non-full bin or not. If the $j$th fresh hash has been checked to go to a non-full bin, then the algorithm will simply terminate with that reveal. So to get an analogous sentence to the above, we need to bound the probability that a $(j+1)$st fresh hash exists, given that the $j$th fresh hash was previously checked to go to a full bin.

We have at most $\delta\epsilon\ell m/5$ hashes outside our table pointing to a bin currently not full and the same number of bins with two pure first balls whose second hashes go to a non-full bin whenever we have $\epsilon\ell m$ slots unoccupied. Similarly, Lemma \ref{boundingepstilde} gives that for this sequence of balls we have $\tilde\epsilon\ge\mb{E}[\tilde\epsilon]-m^{-.49}$. Then we have $\epsilon$ empty slots, our algorithm hasn't failed, and $\tilde\epsilon\ge\mb{E}[\tilde\epsilon]-m^{-.49}\ge\mb{E}[\epsilon^*]-m^{-.49}$. Using that bound on $\tilde\epsilon$ means that there must be at least $(\epsilon-\mb{E}[\tilde\epsilon])\ell m-\ell m^{.51}\ge 3(\epsilon-\mb{E}[\tilde\epsilon])\ell m/4\ge 3(\epsilon-\mb{E}[\epsilon^*])\ell m/4$ remaining (unchecked and unrevealed) hashes that point to a slot that is currently empty. By the lemmas discussed, at most $3(.99^\ell)\epsilon\ell m/5\le 3(\epsilon-\mb{E}[\epsilon^*])\ell m/5$ of those are not in the table or share a bin with another such hash.

This means that there are at least $(\epsilon-\mb{E}[\epsilon^*])\ell m/10$ pure first balls in full bins of our table whose second hash goes to an empty slot and who are alone in their bin, and therefore at least $(\epsilon-\mb{E}[\epsilon^*])\ell m/10$ full bins containing such a ball. Every time we land on one of those bins, by point (4.)~we will finish in just one more eviction.

Every time we perform a special eviction whose hash has been checked to be full, we have probability \[\ge\frac{(\epsilon-\mb{E}[\epsilon^*])\ell m/10}{\#\text{full bins}}\ge \ell(\epsilon-\mb{E}[\epsilon^*])/10\]of landing on a non-corrupt, full bin containing a pure first ball whose second hash goes to an empty slot, and thus there is probability at least $\ell(\epsilon-\mb{E}[\epsilon^*])/10$ that a $(j+1)$st fresh hash exists, given that the $j$th fresh hash was previously checked to go to a full bin.

Then we see that the rest of the proof finishes in the same way, but with $\Pr[b_j\text{ exists}]$ being at most $(1- \ell(\epsilon-\mb{E}[\epsilon^*])/10)^{j+1}$ instead of $(1-\epsilon)^{j+1}$, giving the desired result. Note that the line ``the distribution of $b_j$ (conditioned on it existing) is uniformly random across all non-corrupt bins in $G$'' is still true, if $b_j$ has also been checked to be full, as all bins that were non-full at $\tau_1$ were deemed corrupt.
\end{proof}

Finally, putting the pieces together, we prove the main theorem of the section. Note that Theorem \ref{secondalgthm} proves a slightly stronger runtime bound (by a factor of $\ell$) than in Theorem \ref{thmfirst}.

\secondalgthm*


\begin{proof}
The expected number of evictions up until $\tau_1$ occurs is $O(\epsilon^{-1})$ by Theorem \ref{thmfirst}, so we only have to worry about the run time after $\tau_1$.

Lemma \ref{epsstarsecond} gives that with high probability in $m$, the stopping time of our algorithm (that is, when $(1-\epsilon)\ell m$ balls are inserted for $\epsilon=(1+\delta)(\mb{E}[\epsilon^*])$) occurs before $\tau_2$ occurs. In the case that Lemma \ref{epsstarsecond} does not hold, we can use the trivial run-time bound of $3\ell m$ per ball, as discussed after the definitions of our two algorithms and in Lemma \ref{bicycliconlybarrier}. Then $3\ell m(O(\ell^{-6}m^{-1})=O(\ell^{-5})$, so the event occurring but the algorithm not failing only adds a negligible amount to our expected number of evictions. Similar logic means that we can then assume Lemmas \ref{fewoutsidehashed} and \ref{fewdoubles} hold at every point when $\tau_1$ has occurred, and that the number of non-full bins when $\tau_1$ occurs is between $.91^\ell m$ and $.95^\ell m$ by Lemma \ref{fewemptyattau1}.

Then Lemma \ref{runtimesecondalg}, along with Lemmas \ref{corruptcomp2} and \ref{polylog2}, gives (as in the start of the proof of Lemma \ref{firstalgruntime}) that our final expected number of evictions at any point after $\tau_1$ has occurred is $O((\epsilon-\mb{E}[\epsilon^*])^{-1}\ell^{-1})\le O(\delta^{-1}(\mb{E}[\epsilon^*])^{-1}\ell^{-1})$ as desired.

If $\epsilon\le .9^\ell$, then the number of non-full bins is at most $.9^\ell\ell m\le .91^\ell m$, which by Lemma \ref{fewemptyattau1}means that $\tau_1$ has already occurred. If $\epsilon\ge .9^\ell$ but $\tau_1$ has already occurred, then $O((\epsilon-\mb{E}[\epsilon^*])^{-1}\ell^{-1})$ is $O(\epsilon^{-1})$ as desired (by Corollary \ref{boundingepsstar}, which has now been proven from Lemmas \ref{adapted2} and \ref{epsstarsecond}).
\end{proof}

\section{Query Time}\label{querysection}
While the insertion time of the bucketized cuckoo hash table has been the primary purpose of this paper, we will now turn our attention to the query time.

Assuming that the insertion algorithm has found a valid assignment of objects to bins, then each query must look at at most two buckets. For positive queries, a random choice of the two buckets will find an object in an expected 1.5 bucket checks, as it has a 50/50 chance of guessing the first bucket correctly. This section aims to adapt our algorithm to bring that 1.5 down to $1+o(1)$ as $\ell\rightarrow\infty$.

To do this, we define a third hash function, $p:U\rightarrow(0,1)$, that maps each object to a uniformly random number in the interval $(0,1)$. We think of this as giving the object an eviction priority, and whenever there are multiple balls that we are considering evicting from some bin, we will ``break ties'' by prioritizing the eviction of the ball with the lowest $p$-value. 

In order to guess which bin to check for a given query, we count the number of hashes that have been revealed, which will be $(1+r)\ell m$ for some $r\in[-1,.6]$. Given a query for an object $x$, if $p(x)<r$ we first search for $x$ in $h_2(x)$, while if $p(x)\ge r$ (including if $r<0$) we first search in $h_1(x)$. The hash table algorithm in Subsection \ref{prioritizedalgdef}, along this query algorithm, obtains the guarantee of the following theorem:
\begin{restatable}{theorem}{prioritythm}\label{prioritythm}
Assume $\ell\le 1.5\ln(m)$. There is an algorithm that obtains all properties given in Theorem \ref{secondalgthm}, while also having the property that for any $x$ inserted in our table, the expected number of bins to look in before finding $x$ is $1+o(1)$ as $\ell\rightarrow\infty$.
\end{restatable}

\subsection{``Prioritized'' Algorithm Definition}\label{prioritizedalgdef}

Concretely, we modify the algorithms as follows. Before $\tau_1$ occurs, we modify points (1.)~and (4.) of the initial algorithm:

\begin{enumerate}
    \item If $B$ has an empty slot, place $x$ into the first empty slot of $B$.
    \begin{itemize}
        \item If this ``first empty slot'' is the $\ell$th slot (so the insertion of $x$ makes $B$ full) and there is a ball in $B$ that is on its first hash, then exchange the positions within $B$ of $x$ and the lowest priority ball in $B$ that is on its first hash (which may just be $x$ itself, in which case we leave it).
    \end{itemize}
    \item If the first $\ell-1$ slots of $B$ all contain objects on their second hash, place $x$ into the last slot in $B$.
    \item If $x$ was just evicted from its slot under point (2.)~of this list and $x$ has previously been placed into $B$, place $x$ into the last slot in $B$.
    \item Otherwise, there must be balls in $B$ that are on their first hash. Out of all ``pure first'' balls in $B$, evict the one with minimal priority and place $x$ into its place.
\begin{itemize}
        \item If ``its place'' is the $\ell$th slot and there is another ball in $B$ that is on its first hash, then exchange the positions within $B$ of $x$ and the lowest priority ball in $B$ that is on its first hash (which may just be $x$ itself, in which case we leave it).
\end{itemize}
\end{enumerate}

Similarly, after $\tau_1$ occurs, for the one-step look-ahead algorithm, we modify points (1.), (4.), and (5.):
\begin{enumerate} 
    \item If $B$ has an empty slot, place $x$ into the first empty slot of $B$.
    \begin{itemize}
        \item If this ``first empty slot'' is the $\ell$th slot (so the insertion of $x$ makes $B$ full) and there is a ball in $B$ that is on its first hash, then exchange the positions within $B$ of $x$ and the lowest priority ball in $B$ that is on its first hash (which may just be $x$ itself, in which case we leave it).
    \end{itemize}
    \item If the first $\ell-1$ slots of $B$ all contain objects on their second hash, place $x$ into the last slot in $B$.
    \item If $x$ was just evicted from its slot under point (2.)~of this list and $x$ has previously been placed into $B$, place $x$ into the last slot in $B$
    \item If there are any balls in $B$ that are on their first hash and whose second hash goes to a bin with an empty slot, then evict the lowest priority one of those balls and place $x$ into its place.
\begin{itemize}
        \item If ``its place'' is the $\ell$th slot and there is anothers ball in $B$ that is on its first hash, then exchange the positions within $B$ of $x$ and the lowest priority ball in $B$ that is on its first hash (which may just be $x$ itself, in which case we leave it).
\end{itemize}
    \item Otherwise, there must be balls in $B$ that are on their first hash. Out of those first hash balls, evict the one with minimal priority and place $x$ into its place.
\end{enumerate}
Then, as in Section \ref{secondalgthm}, we choose which one to run based on whether $\tau_1$ has or has not occurred ($\tau_1$ continues to be defined as in Definition \ref{firsttosecondalgdef}). To maintain the property that, every time we check or reveal a hash to a non-full bin, we always put that ball in the bin, we assume that we start checking the second hash for the pure first balls starting with the lowest priority pure first ball, and stop checking if we find a ball whose second hash is non-full.

\subsection{Proving Theorem \ref{prioritythm}}
\begin{restatable}{lemma}{prioritylem}\label{prioritylem}
Let $x$ be a ball with priority $p(x)$. Assume that $(1+r)\ell m$ hashes have been revealed for some $r\in(0,.6)$. Then if $x$ is in our table, \begin{align*}p(x)&<r-\ell^{-.4}\\&\implies x\text{ is in }h_2(x)\text{ with probability at least }1-o(1)\text{ as }\ell\rightarrow\infty.
\\\text{On the other hand, }&r+\ell^{-.4}<p(x)
\\&\implies x\text{ is in }h_1(x)\text{ with probability at least }1-o(1)\text{ as }\ell\rightarrow\infty.
\end{align*}
\end{restatable}
Lemma \ref{prioritylem} (along with Lemma \ref{lem:lowestpriority}) gives the following corollary:
\prioritythm*

\begin{proof}[Proof of Theorem \ref{prioritythm} from Lemma \ref{prioritylem}]
Naively, there is only probability at most $2\ell^{-.4}=o(1)$ that $x$ does not satisfy one of the two conditions in Lemma \ref{prioritylem}. 

(To make the previous sentence fully rigorous, we have to show that $p(x)$ cannot affect the end value of $r$ (assuming the algorithm stopping time is based on inserting a fixed number of elements). E.g., one could be worried that if many $p(x)$ happen to cluster around a point, $r$ may also be more likely to be near that point. To overcome this issue, it is helpful for the analysis to assume that there are two random tapes of uniformly random bins in $[m]$: the first $R_1$ that is read from any time a first hash is revealed, and the second $R_2$ that is read from any time a second hash is revealed. Of course, after $\tau_1$, $R_2$ also needs to be adjusted to account for the checking of hashes. The important thing is that it remains independent from $R_1$. With these two tapes, we see that when an element is about to be evicted from $h_1(x)$, whether it is $x$ or not, will get the same second hash (the next entry of $R_2$). Therefore, $p(x)$ has no effect on the value of $r$ at the algorithm's termination, as it is essentially just an ordering of the balls within a given bin.)

If $x$ satisfies the first condition, there is only an $o(1)$ probability that $x$ is not in the first bin that we check. If $x$ satisfies the second condition, there is also only an $o(1)$ probability that $x$ is not in the first bin that we check. Therefore, the total probability that $x$ is not in the first bin we check is $o(1)+o(1)+o(1)=o(1)$. Since we check at most two bins, this gives the guarantee of $1+o(1)$ expected bins per positive query.

The other properties given in Theorem \ref{secondalgthm} go through in the exact same way as in its proof in Section \ref{secondalgsection}, as this minor modification to this algorithm does not affect the proof of any lemma. In fact, we are still preforming the same algorithm, and are now just specifying a tiebreak proceedure for some decisions that could have been made arbitrarily in Section \ref{secondalgsection}.
\end{proof}
To simplify the notation in our analysis, we will adopt the following definition
\begin{definition}
For any $r\in[-1,1]$, let \underline{$T(r)$} be the event that $(1+r)\ell m$ hashes have been revealed.
\end{definition}

We will also refer to $T(r)$ as a ``time'' and will talk about events likely to occur when $T(r)$ happens. For this section, our notion of likely will be ``with high probability in $\ell$'', that is, for any $c\in\mb{N}$, there is a $C$ (depending on $c$) such that the probability is at least $1-C\ell^{-c}$.

For this, we will implicitly assume that, if $T(r)$ does not occur (due to the algorithm failing), the event in question does occur. For $r\le .6$, this assumption is justified by Lemma \ref{adapted2}, which says that there is probability $O(\ell^{-6}m^{-1})$ that $T(r)$ does not occur.

We observe that most first hashes are revealed before most second hashes:

\begin{lemma}\label{fewsecondhashes}
With high probability in $m$ over the insertion process, we have that at $T(-\ell^{-.49})$, at most $e^{-\ell^{.01}}m$ second hashes have been revealed.
\end{lemma}
\begin{proof}
It is again helpful for the analysis to assume that there are two random tapes of uniformly random bins in $[m]$: the first $R_1$ that is read from any time a first hash is revealed, and the second $R_2$ that is read from any time a second hash is revealed. (Note that $T(-\ell^{-.49})$ occurs before the second algorithm, so we do not have to worry about checking hashes before they are revealed.) Then for a bin $B$, let $f(B)$ be the number of times that $B$ appears in the first $(\ell-\ell^{.51})m$ positions of $R_1$. It is clear that the ``overflow'' \[
Ov:=\sum_B\max(f(B)-\ell, 0)
\] is an upper bound on the number of balls that can be evicted by $T(-\ell^{-.49})$. Therefore, $Ov$ is also an upper bound on the number of second hashes revealed at $T(-\ell^{-.49})$. We see that\begin{align*}
\mb{E}[Ov]&=\mb{E}(\sum_B\max(f(B)-\ell, 0)) = m\mb{E}\max(f(B)-\ell, 0))
\\&=m\sum_{j=1}^\infty j\mb{P}(f(B)=\ell+j)=m\sum_{j=1}^\infty\mb{P}(f(B)\ge\ell+j)
\\&=m\sum_{j=1}^\infty\mb{P}(\mathrm{Bin}((\ell-\ell^{.51})m,1/m)\ge\ell+j)
\\&\le m\sum_{j=1}^\infty e^{-\ell((1+j/\ell)-(1-\ell^{-.49}))^2/2}\say{by Lemma \ref{genpoisson}}
\\&\le m\sum_{j=1}^\infty e^{-\ell(j^2/(\ell^2)+\ell^{-.98})/2}=me^{-\ell^{.02}/2}\left(\sum_{j=1}^\infty e^{-j^2/\ell}\right)
\\&\le me^{-\ell^{.02}/2}(\ell^3)
\end{align*}

Note that changing any one value in $R_1$ can change $Ov$ by at most 1. So, McDiarmid's Bounded Difference Inequality then give us that \[
\mb{P}(Ov\ge e^{-\ell^{.01}m}\le\mb{P}(Ov\ge 2m\ell^3e^{-\ell^{.02}/2})\le\mb{P}(Ov-\mb{E}[Ov])\ge m\ell^3e^{-\ell^{.02}/2}\ell^3)\le e^{\frac{(m\ell^3e^{-\ell^{.02}/2}))^2}{(\ell-\ell^{.51})m}},
\]
giving us a with high probability in $m$ statement as desired, noting that $\ell\le 1.5\ln(m)$.
\end{proof}

Given Lemma \ref{fewsecondhashes}, it is natural to assume that most bins only have first hashes at $T(-\ell^{-.49})$, and mostly receive second hashes after that. We formalize this idea in the following lemma.
\begin{lemma}\label{Tminus51}
Let $x$ be a ball, and let $B$ be a bin, possibly with $B$ defined to be either $h_1(x)$ or $h_2(x)$. With high probability in $\ell$, the following are true: \begin{enumerate}
\item Prior to time $T(-\ell^{-.49})$, we have that $B$ receives only first hashes.
\item By time $T(-\ell^{-.49})$, we have that $B$ has received at least $\ell - \ell^{.52}$ and at most $\ell$ first hashes.
\item Across all time after $T(-\ell^{-.49})$, we have that $B$ receives at most $\ell^{.52}$ total additional first hashes.\end{enumerate}
    \label{lem:numfirsthashes}
\end{lemma}
\begin{proof}
It is again helpful to assume that first hashes are read off a random tape $R_1$ and second hashes are read off a random tape $R_2$.

Then assuming Lemma \ref{fewsecondhashes} holds (which happens with high probability in $m$), we have that $B$ only will receive a second hash if it appears in the first $e^{-\ell^{.01}}m$ slots in $R_2$. Since $B$ has only a $1/m$ probability of appearing in any slot, a union bound tells us that the probability of $B$ appearing in the first $e^{-\ell^{.01}}m$ slots of $R_2$ is $e^{-\ell^{.01}}m$, which is indeed $O(\ell^{-.4})$.

We see from this two-tape analysis that $B$ being set equal to $h_1(x)$ does not break the logic in the previous paragraph, so we have proven (1.) if $B$ is $h_1(x)$ or a non-conditioned bin.

If $B$ is set to equal $h_2(x)$, we can now assume that all other second hashes are read off of $R_2$, but the second hash of $x$ is set to equal $B$. Therefore, we also need to be worried about the probability that $x$ is evicted. $x$ can only be evicted if $h_1(x)$ fills up by $T(-\ell^{.52})$. That can only happen if $h_1(x)$ appears at least $\ell-1$ additional times between the first $(\ell-\ell^{.51})$ slots of $R_1$ and the first $e^{-\ell^{.01}}m$ slots of $R_2$. The probability of this is at most \[\mb{P}(\mathrm{Bin}((1-\ell^{-.49}+e^{-\ell^{.01}})\ell m,1/m)\le 2e^{-\ell(\ell^{-.49}-e^{-\ell^{.01}})^2/2}\le e^{-\ell^{.01}}\]
by Lemma \ref{poisson}. This completes the proof that (1.) holds with high probability in $\ell$.

Assuming that Lemma \ref{fewsecondhashes} holds, for the lower bound in statement (2.), it is sufficient to show that $B$ appears at least $\ell - \ell^{.52}$ times in the first $(\ell-\ell^{.51}-e^{-\ell^{.01}})m$ slots of $R_1$. (The previous sentence is still true if we condition on $B=h_1(x)$ or $B=h_2(x)$.) The probability of this not happening is \[
\mb{P}(\mathrm{Bin}((\ell-\ell^{.51}-e^{-\ell^{.01}})m, 1/m)\le \ell-\ell^{.52})\le e^{-\ell(\ell^{-.48}-\ell^{-.49}-e^{-\ell^{.01}})^2/2}\le e^{-\ell^{.001}/2}
\]by Lemma \ref{genpoisson}. Conditioning on $B=h_2(x)$ has no effect on this, and conditioning on $B=h_1(x)$, similarly to above, can allow us to read all other first hashes (except that of $x$) off of $R_1$, and the above equation still holds.

For the upper bound in Statement (2.), it suffices to show that $B$ appears at most $\ell$ times in the first $\ell-\ell^{.51}$ slots of $R_1$ (or at most $\ell-1$ times if $B=h_1(x)$). Then \[\mb{P}(\mathrm{Bin}((1-\ell^{-.49})\ell m,1/m)\ge \ell-1)\le (1-\ell^{-.49})^{-1}e^{-\ell(\ell^{-.49})^2/2}\]by Lemma \ref{poisson}, giving us a with high probability in $\ell$ statement as desired.

Finally, assuming that Lemma \ref{fewsecondhashes} holds, for statement (3.) to hold it suffices to show that $B$ appears at most $\ell^{.52}$ times in the final $(\ell^{.51}+e^{-\ell^{.01}})m$ indices in $R_1$ (taking $R_1$ to have length $\ell m$). The probability that this does not hold is then \[
\mb{P}(\mathrm{Bin}((\ell^{.51}+e^{-\ell^{.01}})m, 1/m)\ge \ell^{.52})\le e^{-\ell(\ell^{-.48}-\ell^{-.49}-e^{-\ell^{.01}})^2/2}\le e^{-\ell^{.001}/2}
\]as desired, again by Lemma \ref{genpoisson}. Note that conditioning on $B=h_2(x)$ has no effect here, and conditioning on $B=h_1(x)$ may increase its number of first hashes by at most one, so the proof still goes through but with $\ell^{.52}$ replaced with $\ell^{.52}-1$.
\end{proof}
The next natural extension, now that we have bounded the number of first hashes to $B$, is to bound the number of second hashes to $B$, and extend these to any times in our process. This is done in the following lemma.
\begin{lemma}
Let $x$ be a ball, and let $B$ be a bin, possibly with $B$ defined to be either $h_1(x)$ or $h_2(x)$. With high probability in $\ell$, the following is true: At all times until $\tau_2$ occurs, if $r$ is the current value of $r$ at that time, then $B$ has received $\ell + r \ell \pm 2\ell^{.52}$ total hashes.    
    \label{lem:numhashes}
\end{lemma}
\begin{proof}
Throughout this lemma, we will assume that Lemmas \ref{fewsecondhashes} and \ref{Tminus51} hold, and again consider the random tapes $R_1$ and $R_2$.

Note that we can assume $r$ is an integral multiple of $1/(\ell m)$ for $T(r)$ to be well-defined. Furthermore, it then suffices to prove this for those $r$ such that $\lfloor r\ell+2\ell^{.52}\rfloor$ is one lower than $\lfloor (r+1/(\ell m))\ell+2\ell^{.52}\rfloor$ and those $r$ such that $\lfloor r\ell-2\ell^{.52}\rfloor$ is one higher than $\lfloor (r-1/(\ell m))\ell-2\ell^{.52}\rfloor$, as if the lemma fails for any $r$, it must fail for such an $r$. This gives us that we only have $O(\ell)$ values of $r$ that we are interested in proving this statement for. Since we are looking for a ``with high probability in $\ell$'' statement, this means that it suffices to prove that the with high probability statement holds for any particular $r$, as that then overcomes this $O(\ell)$ union bound.

Case 1: $0\le r\le .4$.  By points (2.) and (3.) of Lemma \ref{Tminus51}, it must be true that the number of first hashes that $B$ has received at $T(r)$ is $\ell\pm O(\ell^{.52})$. By Lemma \ref{fewsecondhashes}, the number of second hashes revealed is at least $rm$ and at most $r\ell m+e^{-\ell^{.01}}m$. Then the probability that $B$ receives fewer than $r-\ell^{.52}$ second hashes is at most \begin{equation}
    \mb{P}(\mathrm{Bin}(r\ell m,1/m)\le r\ell-\ell^{.52})\le e^{-\ell(\ell^{-.48})^2/2}\label{lowerr}
\end{equation}by Lemma \ref{genpoisson}, giving us that with high probability in $\ell$, $B$ receives at least $r-\ell^{.52}$ second hashes. Similarly, the probability that $B$ receives more than $r+\ell^{.52}$ second hashes is at most \begin{equation}
\mb{P}(\mathrm{Bin}((r+\ell^{-1}e^{-\ell^{.01}})\ell m,1/m)\ge r\ell+\ell^{.52})\le e^{-\ell(\ell^{-.48}-\ell^{-1}e^{-\ell^{.01}})^2/2}\label{upperr}
\end{equation} again by Lemma \ref{genpoisson}, giving us that with high probability in $\ell$, $B$ receives at most $r+\ell^{.52}$ second hashes.

These both remain true if $B$ is conditioned to equal $h_1(x)$, as they both only deal with the random tape $R_2$. If $B$ is conditioned to equal $h_2(x)$, then we can now assume that all other second hashes are read off of $R_2$, but the second hash of $x$ is set to equal $B$, and then the same proof goes through but with equation \eqref{lowerr} altered to $\mb{P}(\mathrm{Bin}(rm-1,1/m)\le r-\ell^{.52})$ for the lower bound on number of second hashes and equation \eqref{upperr} altered to $\mb{P}(\mathrm{Bin}((r+e^{-\ell^{.01}})m,1/m)\ge r+\ell^{.52})-1$ for the upper bound on number of hashes.

Case 2: $r>.4$. We have shown that with high probability in $\ell$, $B$ has at least $1.4\ell-O(\ell^{.52})$ hashes at $T(.4)$, so $B$ does not have any empty slots remaining at $T(.4)$. Now, we must also consider that the second hashes revealed will have also been checked. Assume that there is a third source of randomness, $R_3$, that determines whether a check is to a non-full bin or to a full bin. Then, say that $R_2$ is only read from when a second hash to a full bin is revealed, and there is a fourth source of randomness $R_4$ that determines the revealed for second hashes that were checked to go to a non-full bin. $R_2$ still has all bins in it (each with probability $1/m$ in each slot), but when an entry of $R_2$ is read that points to a non-full bin, we simply skip over that entry and read off the next one, continuing from there.

To bound how many entries are skipped, it suffices to show that there is some quantity $q$ such that with high probability, at most $q$ of the entries between positions $.4\ell m-e^{-\ell^{.01}}m$ and $.6\ell m + q$ of $R_2$ are non-full at $\tau_1$. We claim that this is true for $q=m/\ell$. Assume that Lemma \ref{fewemptyattau1} holds. As there are at most $.95^\ell m$ non-full bins at $\tau_1$, we have by standard Chernoff bounds that \[\mb{P}\left(\mathrm{Bin}\left(\left(.2\ell+e^{-\ell^{.01}}+\ell^{-1}\right) m,.95^\ell\right)\ge m/\ell\right)\le \left(\frac{.95^\ell\left(.2\ell+e^{-\ell^{.01}}+\ell^{-1}\right) em}{m/\ell}\right)^{m/\ell}\le e^{-m/\ell}\]for sufficiently large $\ell$, giving us a with high probability in $\ell$ statement as desired.

For the upper bound on the number of hashes to $B$, it then suffices to bound the number of times that $B$ appears in $R_2$ between position $.4\ell m-e^{-\ell^{.01}}m$ and position $r\ell m+m/\ell$. The same proof as in Case 1 then goes through but with Equation \eqref{upperr} altered to $\mathrm{Bin}((r+\ell^{-1}e^{-\ell^{.01}}+\ell^{-2})\ell m,1/m)$.

Then for the lower bound, we also need to lower bound the number of times that $R_2$ is read from. For this, it suffices to note that every time a hash is revealed to a non-full bin, an empty slot gets filled. Lemma \ref{fewemptyattau1} then implies that this can happen at most $.95^\ell m$ times, so the same proof as in Case 1 goes through but with Equation \eqref{lowerr} altered to $\mathrm{Bin}((r-\ell^{-1}(.95)^\ell)\ell m,1/m)$.

Case 3: $r\le 0$. We know by the above with $r=0$ that $B$ (with high probability in $\ell$) receives at most $\ell^{.52}$ second hashes before $T(0)$. Then, regarding first hashes, the proof goes through analogously to case 1, but now analyzing $R_1$, and with \eqref{lowerr} altered to $\mb{P}(\mathrm{Bin}((r-\ell^{-.48})\ell m,1/m)\le r\ell-2\ell^{.52}$ and equation \eqref{upperr} altered to $\mathrm{Bin}(r\ell m,1/m)$.
\end{proof}

We now have shown how the pattern of first and second hashes to an average $B$ evolves over time. The next few lemmas will show how this interacts with our notion of priority. First, a more basic lemma to set the framework:

\begin{lemma}
Let $x$ be a ball, and let $B$ be a bin, possibly with $B$ defined to be either $h_1(x)$ or $h_2(x)$. With high probability in $\ell$, the following are true:\begin{enumerate}
    \item $B$ never becomes corrupt
    \item Every time an element is evicted from $B$, it is the pure-first ball with the lowest priority.
\end{enumerate}
    \label{lem:lowestpriority}
\end{lemma}
\begin{proof}
(1.) holds because by Lemma \ref{lem:numhashes}, at $\tau_2$, $B$ has still only received at most $.6\ell+2\ell^{.52}$ second hashes.

(2.) then holds for times up to $\tau_1$ by the algorithm definition. To show (2.) for times after $\tau_1$, it suffices to show that none of the elements $x$ with $h_1(y) = B$ have bins $h_2(y)$ that, when $\tau_1$ occurs, have free slots. We can also assume that these $h_2(y)$ are unrevealed by $\tau_1$, as otherwise $x$ would be in $h_2(y)$ at $\tau_1$ and could only return to $B$ after $h_2(y)$ is full.

For this, we do not think of $R_1$ and $R_2$, but now instead just take the balls with $h_1(y)=B$ and unrevealed second hashes at $\tau_1$, and look at the chances of any of those being a bin that is non-full at $\tau_1$. By points (2.) and (3.) of Lemma \ref{lem:numfirsthashes}, there are at most $\ell+\ell^{.52}$ first hashes that equal $B$, each of which (by Lemma \ref{fewemptyattau1}) have probability at most $.95^\ell$ of having their second hash go to a bin that is not full at $\tau_1$. A union bound then gives us that the probability that any of them have a second hash to a non-full bin is $(\ell+\ell^{.52})(.95)^\ell$, which shows that with high probability none of them have a second hash to a non-full bin.
\end{proof}
Now, we have two tasks: to show that the balls with sufficiently high priorities are not evicted, and to show that the balls with sufficiently low priorities are evicted. The next two lemmas prove these two sides of the coin. We will start by showing that balls with high priorities are not evicted ``before their time'':
\begin{lemma}
Let $x$ be a ball, and let $B$ be a bin, possibly with $B$ defined to be either $h_1(x)$ or $h_2(x)$. With high probability in $\ell$, it is true that at all times $T(r)$ up until $\tau_2$, $B$ has evicted only elements with priorities at most $r + \ell^{-.47}$.
\label{lem:bunchofcorollaries}
\end{lemma}
\begin{proof}
We first note (as in Lemma \ref{lem:numhashes}) that we only need to union bound over $O(\ell)$ values of $r$ -- here, the values of $r$ such that there is a ball $y$ with $h_1(y)=B$ and $p(y)$ between $r+\ell^{.53}$ and $r+1/(\ell m)+\ell^{.53}$; in other words, the final $r$ values before there is one additional ball below the desired priority. So again it suffices to prove a with high probability in $\ell$ statement for any particular $r>\ell^{.53}$.

At $T(r)$, $B$ has received $\ell + r\ell\pm 2\ell^{.52}$ total hashes (by Lemma \ref{lem:numhashes}), of which all after the first $\ell$ evicted the pure-first ball with the lowest priority (by Lemma \ref{lem:lowestpriority}). Thus, $B$ has performed at most $r\ell+2\ell^{.52}$ evictions.

Lemma \ref{lem:numfirsthashes} tells us that by $T(\ell^{-.49})$ (and thus by $T(r)$ for $r>\ell^{-.47}$), $B$ has received at least $\ell$ first hashes. So, for this lemma to fail, we would need that fewer than $r\ell+2\ell^{.52}$ of those $\ge\ell$ first hashes have priority at most $r+\ell^{-.47}$. The probability of this is \[\mb{P}(\mathrm{Bin}(\ell,r+\ell^{-.47})\le r\ell+2\ell^{.52})\le e^{-\frac{(\ell^{.53}-2\ell^{.52})^2}{2\mb{E}[\mathrm{Bin}(\ell,r+\ell^{-.47})]}}\le e^{-(\ell^{1.05})/(2\ell)}\]by a standard (additive) Chernoff bound, giving that (2.) holds with high probability in $\ell$.
\end{proof}

\begin{lemma}
Let $x$ be a ball, and let $B$ be a bin, possibly with $B$ defined to be either $h_1(x)$ or $h_2(x)$. Then with high probability in $\ell$, it is true that at all times $T(r)$, we have that $B$ does not contain any pure-first ball with priority less than $r-\ell^{-.45}$. \label{lem:evictedfirsthashes}
\end{lemma}
One caveat: there may be lower-priority balls $y$ that enter the hash table later in our process and reveal $h_1(y)=B$, who are then immediately evicted from $B$ themselves, rather than evicting a different ball. As these balls were never assigned a slot in $B$, we do not consider $B$ to contain them.
\begin{proof}
Once again, note that it suffices to prove this with high probability in $\ell$ for any $r>0$. This is because there are $O(\ell)$ things to union bound over: for each ball $y$ with $h_1(y)=B$, the minimal value of $r$ such that $r-\ell^{-.45}>p(y)$.

\textbf{Case 1:} First, consider the balls $y$ that has $h_1(y)$ revealed to be $B$ before $T(-\ell^{-.49})$. We will show the lemma applies to them but with priorities less than $r-\ell^{-.46}$. By point (1.) of Lemma \ref{lem:numfirsthashes}, no ball has yet been evicted from $B$. By points (2.) and (3.) of Lemma \ref{lem:numfirsthashes}, $B$ receives at most $\ell+\ell^{.52}$ first hashes over the insertion process. Lemma \ref{lem:numhashes} tells us that by $T(r)$, there have been at least $\ell+r\ell-2\ell^{.52}$ hashes, of which at least $r\ell-3\ell^{.52}$ are second hashes. Lemma \ref{lem:lowestpriority} says that each of those $r\ell-3\ell^{.52}$ second hashes evicts the lowest-priority ball in $B$.

So for a pure-first ball who was in $B$ by $T(-\ell^{-.49})$ not to be evicted with priority less than $r-\ell^{-.46}$ means that there must be at least $r\ell-3\ell^{.52}$ of the $\le\ell+\ell^{.52}$ balls with priority less than $r-\ell^{-.46}$. The probability of this is \[\mb{P}(\mathrm{Bin}(\ell+\ell^{.52},r-\ell^{-.45})\ge r\ell-3\ell^{.52})\le e^{-\frac{(\ell^{.54} + (3-r)\ell^{.52}+\ell^{.07})^2}{2(r\ell + r\ell^{.52}-\ell^{.54}-\ell^{.07})}}\le e^{\ell^{-.07}}\]using an additive Chernoff bound. So, with high probability in $\ell$, all such balls are evicted.

\textbf{Case 2:} Note that once $B$ is full, the priority of the lowest-priority ball in $B$ can never decrease, as any lower-priority ball that arrives by Lemma \ref{lem:lowestpriority} will itself gets evicted immediately. Therefore, to finish the proof, it suffices to show that there exists a ball with priority at least $r-\ell^{-.45}$ is evicted from $B$ before $T(r)$.

By Lemma \ref{lem:numhashes}, there are at least $\ell^{.54}-3\ell^{.52}$ second hashes that arrive in $B$ between $T(r-\ell^{-.46})$ and $T(r)$. As there are only at most $\ell^{.52}$ first hashes that arrive after $T(-\ell^{-.49})$, some pure-first ball that arrived before $T(-\ell^{-.49})$ is evicted by a second hash that arrives between $T(r-\ell^{-.46})$ and $T(r)$. By Case 1 applied on $T(r-\ell^{-.46})$, that ``early arriving'' ball $z$ had priority at least $r-2\ell^{-.46}>r-\ell^{-.45}$ and was evicted.
\end{proof}

\prioritylem*
\begin{proof}
By Lemma \ref{lem:lowestpriority}, we have with high probability in $\ell$ that neither bin $h_1(x)$ nor $h_2(x)$ is corrupt. So, which bin $x$ is in will depend only on whether it was ever evicted from $h_1(x)$. By Lemma \ref{lem:bunchofcorollaries}, bin $h_1(x)$ has (with high probability in $\ell$) only evicted pure-first elements with priorities $\le r + \ell^{-.47}$. By Lemma \ref{lem:evictedfirsthashes}, it has evicted all pure-first elements with priorities $\ge r  - \ell^{-.45}$. So, the only case we must worry about is if $p(x) \in [r - \ell^{-.45}, r + \ell^{-.47}]$. 

\end{proof}

\section{Comparison to Random Walk}\label{sec:separation}
The random walk insertion procedure, when landing on a bin $B$, chooses which ball to evict from $B$ uniformly at random. In this section, we show that there is a sense in which random walk insertion is provably worse than our first insertion algorithm. In particular, we prove the following theorem:

\begin{restatable}{theorem}{rwbad}\label{rwbad}
Assume $\ell\le 1.5\ln(m)$. Fix $M$ such that $3\le M\le(\ln\ell)/20$ and assume that the random walk procedure has performed at most $M\ell m$ evictions over the insertion process. Then with high probability in $m$, we have that there exists an $C=C(M)>0$ such that the current load factor is at most $1-(2/e)^{\ell-C\ell}$
\end{restatable}

Recalling that $\epsilon^* = (2/e)^{\ell \pm  o(\ell)}$, this immediately gives us the following corollary:
\begin{corollary}
Assume $\ell\le 1.5\ln(m)$. With high probability in $m$, to get to load factor $(1 - \epsilon)$ with $\epsilon \le e^{o(\ell)} \epsilon^* = {\epsilon^*}^{(1 - o(1))}$, the random-walk algorithm must perform $\omega(\ell m)$ evictions as $\ell\rightarrow\infty$. 
\end{corollary}
Comparing this to Theorem \ref{thmfirst}, we see that our algorithm reaches load factor $1 - \epsilon$ with $\epsilon = \epsilon^* e^{o(\ell)} = {\epsilon^*}^{(1 - o(1))}$ in $O(\ell m)$ evictions (with probability $1-O(\ell^{-6}m^{-1})$), giving a separation between the two algorithms.

Our proof still holds for the variant of random walk insertion where we first check whether the bin we are evicting from has any empty slots, and if so insert into one of those.

\subsection{Proving Theorem \ref{rwbad}}

Essentially, the barrier will be a ``coupon collector problem'': that a bin with $\ell$ slots on average needs to be evicted from $\Theta(\ell\log\ell)$ times during the random walk insertion in order for every slot in that bin to be evicted from. If the random walk algorithm performs $M\ell m$ evictions for some constant $M$, then each bin performs on average $M\ell$ evictions. Since $M\ell\ll\ell\log\ell$, there will there will likely be some slots in each bin that were never evicted from. These slots are likely to contain balls that have never seen their other hash. We then finish the proof by showing that, with so many hashes unrevealed, it is highly unlikely to reach as high of a load factor as our algorithm did.

For our analysis, it is useful to imagine that each bin $B$ has its own random tape, that is, that $B$ has a random function $r_B:\mb{N}\rightarrow[\ell]$, and when visited for the $i$-th time, evicts the object from the $r_B(i)$-th slot. The following lemma formalizes the ``coupon collector'' issue, by showing that most bins have a decent number of slots, $.5e^{-2M}\ell$ many, that won't be visited in the first $2\ell M$ visits to that bin.
\begin{lemma}\label{binrandomness}
Fix $M$ such that $1\le M\le(\ln\ell)/4$. Let $T$ be the number of bins $B$ such that \[|\{r_B(i)\mid i\in[2\ell M]\}|>\ell (1-.5e^{-2M}).\] Then with high probability in $m$, we have $T\le m/5$.
\end{lemma}
\begin{proof}
For any given bin $B$, let $c_B=\ell-|\{r_B(i)\mid i\in[2\ell M]\}|$, that is, the number of unvisited slots in the first $2\ell M$ visits to $B$. We have that \begin{align*}
\mb{E}[c_B]=\ell\left(1-\frac 1\ell\right)^{2\ell M}\le \ell e^{-2M}.
\end{align*}Then by standard Chernoff bounds, for a given bin $B$, we have\[
\mb{P}(c_B\le .5\ell e^{-2M})\le\mb{P}(c_B\le .5\mb{E}[c_B])\le e^{-\ell e^{-2M}/8}.
\]Then noting that $T$ is exactly the set of bins $B$ where $c_B\le .5\ell e^{-2M}$, we have $\\\mb{E}[T]\le me^{-\ell e^{-2M}/8}\le me^{-\sqrt{\ell}/8}\le m/10$, so again by Chernoff bounds\[
\mb{P}(T\ge m/5)\le\mb{P}(T\ge \mb{E}[T]+m/10)\le e^{-m/20}
\]as desired.
\end{proof}
We now use this coupon collector problem to show that we will have $\Theta(\ell m)$ unrevealed hashes after $M\ell m$ evictions:
\begin{lemma}\label{stillpurefirst}
Fix $M$ such that $3\le M\le(\ln\ell)/20$. and assume that the random walk procedure has performed at most $M\ell m$ evictions. With high probability in $m$, at most $(2-e^{-5M})\ell m$ hashes have been revealed.
\end{lemma}
\begin{proof}
We will here use ``first hashes'' to mean the first hash that an object reveals during its initial bin placement (which can be randomized in random walk insertion).

If fewer than $(1-e^{-5M})\ell m$ first hashes have been revealed, then as there are at least as many first hashes revealed as second, the lemma is true. Therefore, we can assume that at least $(1-e^{-5M})\ell m$ first hashes have been revealed. Then the number of first hashes to a given bin stochastically dominates by $\mathrm{Bin}((1-e^{-5M})\ell m,1/m)$, so by Lemma \ref{genpoisson} we have that

\[
\mb{P}(\mathrm{Bin}((1-e^{-5M})\ell m,1/m)\ge (1-e^{-4M})\ell)\le e^{-\ell((1-e^{-4M})-(1-e^{-5M}))^2/2}\le e^{-\ell e^{-10M}}
\]
So let $R$ be the number of bins that receive at least $(1-e^{-4M})\ell$ first hashes by the time $(1-.5e^{-5M})\ell m$ first hashes have been revealed. Then $\mb{E}[R]\le me^{-\ell e^{-10M}}\le me^{-\sqrt{\ell}}\le m/10$ and \[
\mb{P}(R\ge m/5)\ge\mb{P}(R\ge \mb{E}[R]+m/10)\le e^{-m/20}.
\]

If at most $M\ell m$ evictions have been performed, then there must be at least $m/2$ bins that have each received at most $2M\ell$ evictions. Assuming $T\le m/5$ (from Lemma \ref{binrandomness}) and $R\le m/5$ (from the paragraph above), we have that there must be at least $m/2-m/5-m/5=m/10$ bins $B$ with at most $2M\ell$ evictions, $|\{r_B(i)\mid i\in[2\ell M]\}|>\ell (1-.5e^{-2M})$, and at least $(1-e^{-4M})\ell$ first hashes. Then the following two hold:\begin{itemize}
    \item At least $(1-e^{-4M})\ell$ first hashes have come to $B$
    \item At least $.5e^{-2M}\ell$ slots in $B$ have not been evicted from
\end{itemize}
Therefore, we see that there are at least $.5e^{-2M}\ell-e^{-4M}\ell\ge e^{-4M}\ell$ slots in $B$ that contain a first hash ball that has never been evicted. Since there are $m/10$ such bins, there are then at least $e^{-4M}\ell m/10\ge e^{-5M}\ell m$ balls who have never revealed their second hash, and thus at most $(2-e^{-5M})\ell m$ total hashes revealed.

\end{proof}
Finally, we show that with this many unrevealed hashes, the load factor is unlikely to reach that of Theorem \ref{thmfirst}.
\rwbad*
In fact, we will prove this for $C=e^{-6M}/4$.
\begin{proof}
Assume Lemma \ref{stillpurefirst} holds. Then there have been at most $(2-e^{-5M})\ell m$ hashes revealed. We upper bound the expected current load factor by performing a similar calculation to Lemma \ref{boundingepstilde}. We count the number of ``leftover'' slots (if $B$ gets $j$ hashes for $j<\ell$, it contributes $\ell-j$ leftover slots). Let $\epsilon$ be the current load factor. Then \begin{align*}
\ell\mb{E}[\epsilon]&\ge (1-O(1/\ell))\sum_{j=0}^\ell(\ell-j)\mb{P}(\mathrm{Bin}((2-e^{-5M})\ell m,1/m)=j)\\&\ge(1-O(1/\ell))e^{-\ell(1-\ln(2)-e^{-6M})}\\&\ge(1-O(1/\ell))(2/e)^{\ell-\ell e^{-6M}/4}
\end{align*}
Finally, we can note that the proof of Lemma \ref{boundingepstilde} still goes through to show that \[
\epsilon\ge (1+O(1/\ell))(2/e)^{\ell-\ell e^{-6M}/4}+m^{-.49}
\] is true with high probability in $m$. Then, since $(2/e)^\ell$ is $\Omega(m^{-.47})$ and thus $(2/e)^{\ell-\ell e^{-6M}/4}\ge(2/e)^\ell=\Omega(m^{-.47})$ as well, we have
\[
\epsilon\ge (1+O(1/\ell))(2/e)^{\ell-\ell e^{-6M}/4}
\]
with high probability in $m$.
\end{proof}

\bibliographystyle{alpha}
{\fontsize{11}{11}\bibliography{cuckoo}}

\appendix
\section{Appendix}\label{app:foo}
\expepstilde*
\begin{proof}
Linearity of expectation tells us that $\ell\mb{E}[\tilde\epsilon]=\mb{E}[\max(0,\ell-(\text{\# hashes to }B))]$ for a given bin $B$. The number of hashes $B$ receives is $\mathrm{Bin}(2\ell m,1/m)$. Therefore,
\begin{align*}\ell\mb{E}[\tilde\epsilon]=\sum_{j=0}^\ell&(\ell-j)\mb{P}(\mathrm{Bin}(2\ell m,1/m)=j)\\&=\sum_{j=0}^{\ell-1}(\ell-j)\binom{2\ell m}{j}(1/m)^j(1-1/m)^{2\ell m-j}
\\&=\sum_{j=0}^{\ell-1}(\ell-j)\frac{(2\ell)^j}{j!}(1-1/m)^{2\ell m-j}\prod_{k=0}^{j-1}\left(1-\frac{k}{2\ell m}\right)
\\&=\sum_{j=0}^{\ell-1}(\ell-j)\frac{(2\ell)^j}{j!}(1-1/m)^{2\ell m-j}\left(1\pm O(\ell/m)\right)
\\&=\sum_{j=0}^{\ell-1}(\ell-j)\frac{(2\ell)^j}{j!}\left(e^{-2\ell}+O(\ell/m)\right)\left(1\pm O(\ell/m)\right)\refstepcounter{equation}\tag{\theequation}\label{poiseq}
\\&=\sum_{j=0}^{\ell-1}(\ell-j)\frac{(2\ell)^j}{j!}\left(e^{-2\ell}\pm O(1/\ell)\right)\\&\tag{Since $j\le\ell$ and $O(\ell/m) \le O(1/\ell) = o(1)$}
\\&=(1\pm O(1/\ell))\sum_{j=1}^{\ell}j\frac{(2\ell)^{\ell-j}}{e^{2\ell}(\ell-j)!}\\&\tag{by reversing the summation order}
\\&=(1\pm O(1/\ell))e^{-2\ell}\frac{(2\ell)^\ell}{\ell!}\sum_{j=1}^\ell\left(\frac{j}{(2\ell)^j}\prod_{k=0}^j(\ell-k)\right)
\\&=(1\pm O(1/\ell))e^{-2\ell}\frac{(2\ell)^\ell}{\ell!}\sum_{j=1}^{\ell}j2^{-j}\\&=(1\pm O(1/\ell))e^{-2\ell}\cdot \frac{2\cdot (2\ell)^\ell}{\ell!}\\&=(1 \pm O(1/\ell)) \cdot \frac{2\cdot (2/e)^\ell}{\sqrt{2\pi\ell}}\tag{by Stirling's formula},
\end{align*}
so \[\mb{E}[\tilde\epsilon]=(1 \pm O(1/\ell)) \cdot \frac{2\cdot(2/e)^\ell}{\ell^{1.5}\sqrt{2\pi}}.\]
We note that Equation \eqref{poiseq} when divided by $\ell$ is $(1 \pm O(1/m)) \cdot\frac 1\ell\cdot \E_{X\sim\operatorname{Poisson}(2\ell)}[\max(0, \ell - X)]$.
\end{proof}

\boundingepstilde*
\begin{proof}
As each hash can change $\tilde\epsilon$ by at most $1/(\ell m)$, McDiarmid's Inequality tells us that \begin{align*}
\mb{P}(|\tilde\epsilon-\mb{E}[\tilde\epsilon]|\ge m^{-.49})&\le 2e^{-\frac{(m^{-.49})^2}{2\ell m((1/(\ell m))^2)}}\le2e^{-m^{.02}\ell/2},
\end{align*}which tends to zero with faster than polynomial convergence for $\ell=O(\log m)$.
\end{proof}

\poisson*
\begin{proof}
For $Y\sim\mathrm{Ber}(1/m)$, we have for any $t\in\mb{R}$ that \[\mb{E}[e^{tY}]=1-\frac 1m+\frac{e^{t}}{m}.\] Then for $X\sim\mathrm{Bin}(c\ell m,1/m)$ for any $c\in\mb{R}_+$ and any $t\in\mb{R}$, we have that \[\mb{E}[e^{tX}]=\left(1-\frac 1m+\frac{e^{t}}{m}\right)^{c\ell m}\le e^{c\ell(e^{t}-1)}.\]Then for any $t>0$ and $c\in(0,1)$, we have\[
\mb{P}(X\ge \ell-1)=\mb{P}(e^{tX}\ge e^{t(\ell-1)})\le e^{-t(\ell-1)}\mb{E}[e^{tX}]\le e^{-t(\ell-1)}e^{c\ell(e^{t}-1)}=e^{\ell(-t+ce^{t}-c)+t}.
\]We optimize this by choosing $t=-\ln(c)>0$, giving for $c\in(0,1)$ that \[\mb{P}(\mathrm{Bin}(c\ell m,1/m)\ge\ell)\le e^{-\ell(c-1-\ln(c))}/c.\]
The final inequality comes by noting that for $c\in(0,1)$ we have $\ln(c)\le (c-1)-(c-1)^2/2$.
Similarly, for $c>1$, for any $t>0$ we have \[
\mb{P}(X\le\ell)=\mb{P}(e^{-tX}\ge e^{-t\ell})\le e^{t\ell}\mb{E}[e^{-tX}]\le e^{t\ell}e^{c\ell(e^{-t}-1)}=e^{\ell(t+ce^{-t}-c)}.
\]We optimize this by choosing $t=\ln(c)>0$, giving for $c>1$ that \[\mb{P}(\mathrm{Bin}(c\ell m,1/m)\le\ell)\le e^{-\ell(c-1-\ln(c))}.\]
The final inequality then comes by noting that for $c\in(1,1.7)$, we have that $\ln(c)\le(c-1)-(c-1)^2/3$.
\end{proof}

\genpoisson*
\begin{proof}
As in Lemma \ref{poisson}, for $X\sim\mathrm{Bin}(c_2\ell m,1/m)$ for any $c_2\in\mb{R}_+$ and any $t\in\mb{R}$, we have that \[\mb{E}[e^{tX}]=\left(1-\frac 1m+\frac{e^{t}}{m}\right)^{c_2\ell m}\le e^{c_2\ell(e^{t}-1)}.\]Then for any $t>0$ and $c_1\in(0,1)$, we have\[
\mb{P}(X\ge c_1\ell)=\mb{P}(e^{tX}\ge e^{c_1t\ell})\le e^{-c_1t\ell}\mb{E}[e^{tX}]\le e^{-c_1t\ell}e^{c_2\ell(e^{t}-1)}=e^{-\ell(c_1t-c_2e^{t}+c_2)}.
\]We optimize this by choosing $t=\ln(c_1/c_2)>0$, giving for $c_1\in(0,1)$ and $c_2\in(0,c_1)$ that \[\mb{P}(\mathrm{Bin}(c_2\ell m,1/m)\ge c_1\ell)\le e^{-\ell(c_1\ln(c_1/c_2)-c_1+c_2)}.\]
The final inequality comes by noting that Taylor's Theorem (viewing $(c_1\ln(c_1/c_2)-c_1+c_2)$ as a function of $c_1$) gives that there is an $\xi\in[c_1,c_2]$ such that\[
c_1\ln(c_1/c_2)-c_1+c_2=\frac{(c_1-c_2)^2}{2\xi}\ge\frac{(c_1-c_2)^2}{2}.
\]
Similarly, for $c_2\in(c_1,1)$, for any $t>0$ we have \[
\mb{P}(X\le c_1\ell)=\mb{P}(e^{-tX}\ge e^{-c_1t\ell})\le e^{c_1t\ell}\mb{E}[e^{-tX}]\le e^{c_1t\ell}e^{c_2\ell(e^{-t}-1)}=e^{-\ell(-c_1t-c_2e^{-t}+c_2)}.
\]We optimize this by choosing $t=\ln(c_2/c_1)>0$, giving for $c_2\in(c_1,1)$ that \[\mb{P}(\mathrm{Bin}(c_2\ell m,1/m)\le c_1\ell)\le e^{-\ell(c_1\ln(c_1/c_2)-c_1+c_2)}.\]
as above. The final line of the equation then again holds as above, noting that we now have $\xi\in[c_2,c_1]$ which is still at most 1.
\end{proof}

\end{document}